\documentclass{acmtrans2m}

\usepackage{amssymb}
\usepackage{graphicx}
\usepackage{xspace}
\usepackage{multirow}

\newtheorem{theorem}{Theorem}[section]

\newtheorem{corollary}[theorem]{Corollary}

\newtheorem{lemma}[theorem]{Lemma}
\newdef{definition}[theorem]{Definition}
\newdef{remark}[theorem]{Remark}

\newcommand{\ioradI}{\textsc{io-radial1}\xspace}
\newcommand{\ioradII}{\textsc{io-radial2}\xspace}

\newcommand{\ioradIII}{\textsc{io-radial3}\xspace}

\newcommand{\cosweep}{\textsc{io-centrifugal}\xspace}

\newcommand{\visdac}{\textsc{vis-dac}\xspace}
\newcommand{\visiter}{\textsc{vis-iter}\xspace}
\newcommand{\rrr}{\texttt{R3}\xspace}
\newcommand{\rr}{\texttt{R2}\xspace}

\newcommand{\io}{I/O\xspace}
\newcommand{\ios}{I/Os\xspace}

\newcommand{\sort}{\mathsf{sort}}
\newcommand{\Sort}[1]{\ensuremath{\frac{#1}B \log_{M/B} \frac{#1}B}}
\newcommand{\scan}{\mathsf{scan}}

\newcommand{\enterevent}{\textsc{enter}\xspace}
\newcommand{\exitevent}{\textsc{exit}\xspace}
\newcommand{\centerevent}{\textsc{center}\xspace}

\newenvironment{denseitems}{\list{$\bullet$}{\itemsep0pt\parsep0pt}}{\endlist}

\title{A Comparison of I/O-Efficient Algorithms for Visibility
  Computation on Massive Grid Terrains}

\author{
HERMAN HAVERKORT\\Technische Universiteit Eindhoven, The Netherlands.\\
LAURA TOMA\\Bowdoin College, Maine, USA.
}

\begin{abstract}
 
  Given a terrain~$T$ and a \emph{viewpoint}~$v$, the \emph{visibility
    map} or \emph{viewshed} of~$v$ is the set of grid points of~$T$
  that are visible from~$v$.  To decide whether a point $p$ is visible
  one needs to interpolate the elevation of the terrain along the
  line-of-sight (LOS) $vp$.  Existing viewshed algorithms differ
  widely in what points they chose to interpolate, how many
  lines-of-sight they consider, and how they interpolate the
  terrain. These choices crucially affect the running time and
  accuracy of the algorithms.
  This paper describes 
  I/O-efficient algorithms for computing visibility maps on massive
  grid terrains in a couple of different models.

  First, we describe two algorithms that use the interpolation model
  of Van Kreveld.
  These algorithms sweep the terrain by rotating a ray around the
  viewpoint while maintaining the terrain profile along the ray. On a
  terrain of $n$ grid points, these algorithms run in $O(n \log n)$
  time and $O(\sort(n))$ \ios in the I/O-model of Aggarwal and
  Vitter. 
  Second, we describe an algorithm which runs in $O(n)$ time and
  $O(\scan(n))$ \ios, and is cache-oblivious. This algorithm sweeps
  the terrain centrifugally, growing a star-shaped region around the
  viewpoint while maintaining the approximate visible horizon of the
  terrain within the swept region.

  Our last two algorithms use linear interpolation and the model of
  Franklin's \rrr algorithm, which in the literature is referred to as
  the ``exact'' algorithm.  Our algorithms are based on computing and
  merging horizons, and we prove that the complexity of horizons on a
  grid of $n$ points with linear interpolation is $O(n)$, improving on
  the general $O(n \alpha(n))$ bound on triangulated terrains.

  We present an experimental analysis of our algorithms on NASA SRTM
  data. All our algorithms are scalable to volumes of data that are
  over 50 times larger than main memory.  Our main finding is that, in
  practice, horizons are significantly smaller than their theoretical
  worst case bound, which makes horizon-based approaches very
  fast. Our last two algorithms, which compute the most accurate
  viewshed, turn out to be very fast in practice, although their
  worst-case bound is inferior.

\end{abstract}

\category{F.2.2} {Analysis of Algorithms and Problem
Complexity}{Nonnumerical Algorithms and Problems} [Geometrical
problems and computations] 

\category{I.3.5}{Computing
Methodologies}{Computational Geometry and Object Modeling} [Geometric
algorithms, languages, and systems]

\terms{Algorithms, Design, Experimentation, Performance}

\keywords{computational geometry, data structures and algorithms,
digital elevation models, I/O-efficiency, terrains, visibility}

\begin{document}

\begin{bottomstuff}
 
A preliminary version of this work appeared in \emph{Proceedings of
the 17th ACM SIGSPATIAL Symp. of Geographic Information Systems (GIS
2009)}. Best-paper award, and in  \emph{Proceedings of
the 21st ACM SIGSPATIAL Symp. of Geographic Information Systems (GIS
2013)}.
\newline
Herman Haverkort, Department of Mathematics and Computer Science, Technische Universiteit Eindhoven, P.O. Box 513, 5600 MB Eindhoven,  The Netherlands, \texttt{cs.herman@haverkort.net}.
\newline
Laura Toma, Department of Computer Science, Bowdoin College, 8650
College Station, Brunswick, ME 04011, USA, \texttt{ltoma@bowdoin.edu}.

\end{bottomstuff}

\maketitle

\section{Introduction}
\label{sec:introduction}

The computation of visibility is a fundamental problem on terrains and
is at the core of many applications such as planning the placement of
communication towers or watchtowers, planning of buildings such that
they do not spoil anybody's view, finding routes on which you can
travel while seeing a lot or without being seen, and computing solar
irradiation maps which can in turn be used in predicting vegetation
cover.  The basic problem is point-to-point visibility: Two points $a$
and $b$ on a terrain are visible to each other if the interior of
their \emph{line-of-sight} $ab$ (the line segment between $a$ and b)
lies entirely above the terrain. Based on this one can define the
viewshed: Given a terrain and an arbitrary (view)point $v$, not
necessarily on the terrain, the \emph{visibility map} or
\emph{viewshed} of $v$ is the set of all points in the terrain that
are visible from $v$; see Figure~\ref{fig:viewshed-ex}.  A variety of
problems pertaining to visibility have been researched in
computational geometry and computer graphics, as well as in geographic
information science and geospatial engineering.

The key in defining and computing visibility is choosing a terrain
model and an interpolation method. The most common terrain models are
the grid and the TIN (triangular irregular network). A grid terrain is
essentially a matrix of elevation values, representing elevations
sampled from the terrain with a uniform grid; the x,y coordinates of
the samples are not stored in a grid terrain, they are considered
implicit w.r.t. to the corner of the grid. A TIN terrain consists of
an irregular sample of points (x,y and elevation values), and a
triangulation of these points is provided. Grid terrains are the most
widely used in GIS because of their simplicity. Our algorithms in this
paper discuss the computation of visibility maps on \emph{grid}
terrains.

To decide whether a point $p$ is visible on a given terrain model, one
needs to interpolate the elevation along the line-of-sight $pv$
between the viewpoint $v$ and $p$ (more precisely, along the
projection of the line-of-sight on the horizontal plane) and check
whether the interpolated elevations are below the line of
sight. Various algorithms differ in what and how many points they
select to interpolate along the line-of-sight, and in the
interpolation method used. These choices crucially affect the
efficiency and accuracy of the algorithms.

In order to be useful in practice, viewshed algorithms need to be fast
and scalable to very large terrains. The last decade witnessed an
explosion in the availability of terrain data at better and better
resolution.  In 2002, for example, NASA's Shuttle Radar Topography
Mission (SRTM) acquired 30~m-resolution terrain data for the entire
USA, in total approximately 10 terabytes of data.  With more recent
technology it is possible to acquire data at sub-meter resolution.
This brings tremendous increases in the size of the datasets that need
to be processed: Washington state at 1~m resolution, using 4 bytes for
the elevation of each sample, would total 689 GiB of data; Ireland
would be 262 GiB---only counting elevation samples on land. Data at
this fine resolution has started to become  available.

\begin{figure}[t]
\centering
\includegraphics[width=6cm]{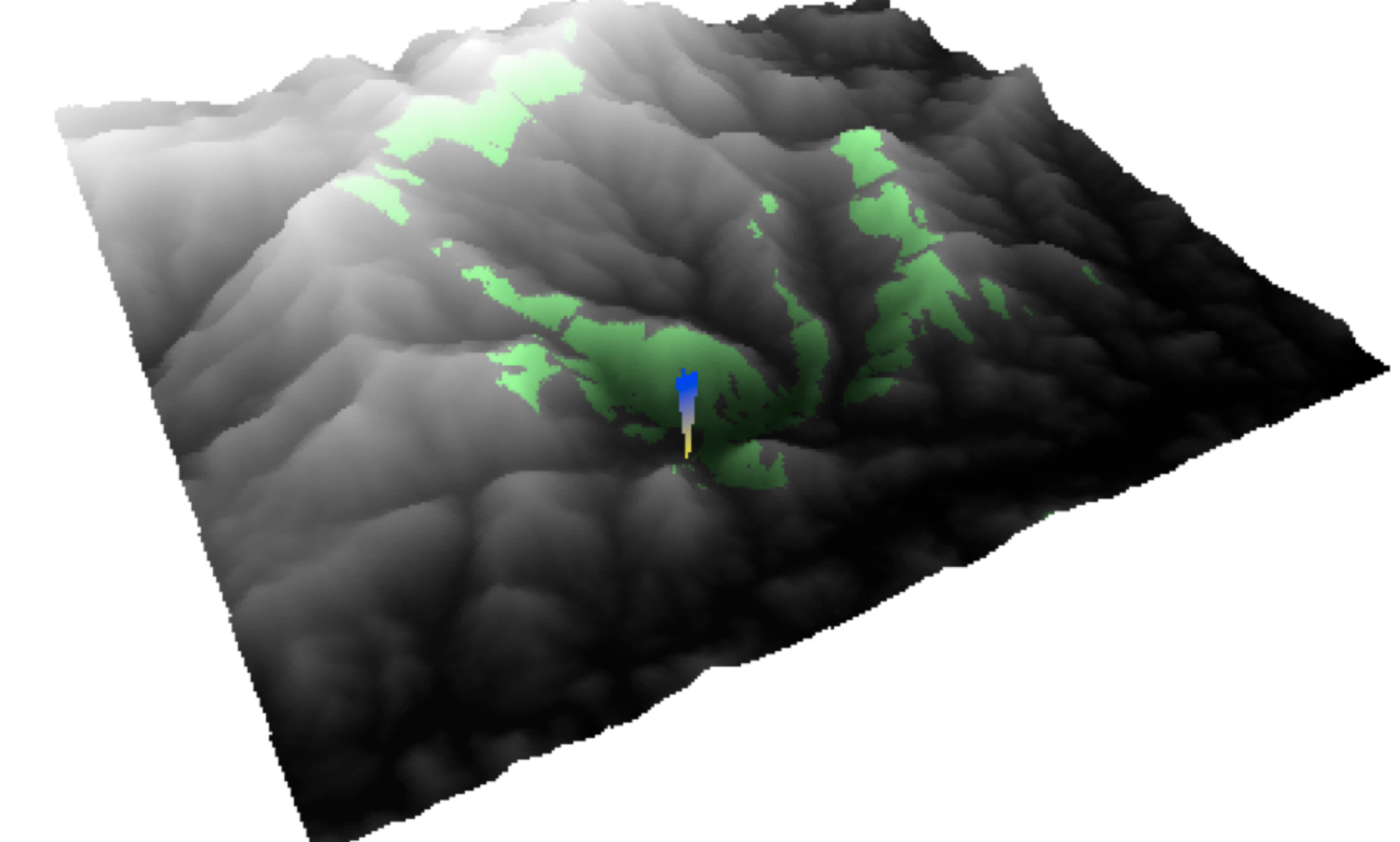}
\includegraphics[width=5cm]{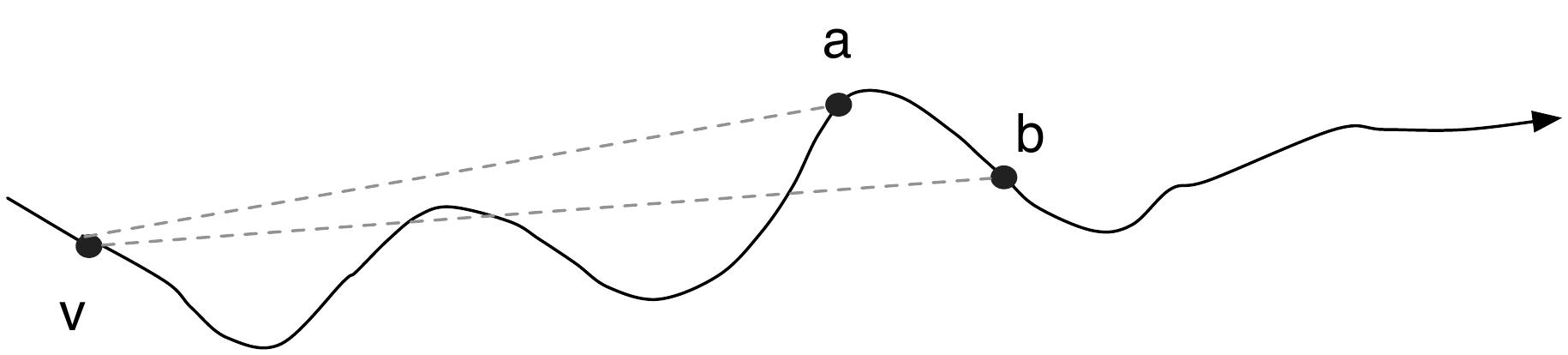}
\caption{(a) The viewshed of a point on a grid terrain is shown in
  green. The viewpoint is marked in blue. (b) Two points on a terrain
  are visible to each other if the interior of their line-of-sight
  lies entirely above the terrain. Here $a$ is visible from $v$, but
  $b$ is not.}
\label{fig:viewshed-ex}
\end{figure}


\subsection{I/O-efficiency}
Working with large terrains require efficient algorithms that scale
well and are designed to minimize ``I/O'': the swapping of data
between main memory and disk.  We assess the efficiency of algorithms
in this paper not only by studying the number of computational steps
they need and by measuring their running times in practical
experiments, but also by studying how the number of I/O-operations
grows with the input size. To this end we use the standard model
defined by Aggarwal and Vitter~\cite{aggarwal:input}. In this model, a
computer has a memory of size $M$ and a disk of unbounded size. The
disk is divided into blocks of size $B$. Data is transferred between
memory and disk by transferring complete blocks: transferring one
block is called an ``\io''. Algorithms can only operate on data that
is currently in main memory; to access the data in any block that is
not in main memory, it first has to be copied from disk.  If data in
the block is modified, it has to be copied back to disk later, at the
latest when it is evicted from memory to make room for another block.
The \io-efficiency of an algorithm can be assessed by analysing the
number of \ios it needs as a function of the input size $n$, the
memory size $M$, and the block size $B$. The fundamental building
blocks and bounds in the I/O-model are sorting and scanning: scanning
$n$ consecutive records from disk takes $\scan(n) = \Theta(n/B)$ \ios;
sorting takes $\sort(n) = \Theta(\Sort{n})$ \ios in the worst
case~\cite{aggarwal:input}. It is sometimes assumed that $M =
\Omega(B^2)$.

We distinguish \emph{cache-aware} algorithms and
\emph{cache-oblivious} I/O-efficient algorithms: Cache-aware
algorithms may use knowledge of $M$ and $B$, (and to some extent even
control $B$) and they may use it to control which blocks are kept in
memory and which blocks are evicted.  Cache-oblivious algorithms, as
defined by Frigo et al.~\cite{prokop:cob}, do not know $M$ and $B$
and cannot control which blocks are kept in memory: the caching policy
is left to the hardware and the operating system.  Nevertheless,
cache-oblivious algorithms can often be designed and proven to be
\io-efficient~\cite{prokop:cob}. The idea is to design the algorithm's
pattern of access to locations in files and temporary data structures
such that effective caching is achieved by any reasonable
general-purpose caching policy (such as least-recently-used
replacement) \emph{for any values} of $M$ and $B$.  As a result, any
bounds that can be proven on the \io-efficiency of a cache-oblivious
algorithm hold for \emph{any} values of $M$ and $B$
simultaneously. Thus they do not only apply to the transfer of data
between disk and main memory, but also to the transfer of data between
main memory and the various levels of smaller caches.  However, in
practice, cache-oblivious algorithms cannot always match the
performance of cache-aware algorithms that are tuned to specific
values of $M$~and~$B$~\cite{brodal:cobsorting-jea}.

\subsection{Problem definition}
A \emph{terrain} $T$ is a surface in three dimensions, such that any
vertical line intersects $T$ in at most one point. The \emph{domain}
$D$ of $T$ is the projection of $T$ on a horizontal plane. The
\emph{elevation angle} of any point $q = (q_x,q_y,q_z)$ with respect
to a viewpoint~$v = (v_x,v_y,v_z)$ is defined as:
$$\textrm{ElevAngle}(q) = \arctan\frac{q_z - v_z}{\textrm{Dist}(q)},$$
where $\textrm{Dist}(q) := |(q_x,q_y)-(v_x,v_y)|$. 

A point $u = (u_x,u_y,u_z)$ is visible from $v$ if and only if the
elevation angle of $u$ is higher than the elevation angle of any point
of $T$ whose projection on the plane lies on the line segment from
$(u_x,u_y)$ to $(v_x,v_y)$.  We define the elevation angle of any
point $(q_x,q_y)$ of $D$ as the elevation angle of the point $q =
(q_x,q_y,q_z)$ where the vertical line through $(q_x,q_y)$ intersects
$T$. 

In this paper we consider terrains that are represented by a set of
$n$ points whose projections on~$D$ form a regular rectangular grid.
To decide whether a point $u$ is visible from a point $v$, we need to
\emph{interpolate} the elevation angle of points of $T$ whose projection
on the plane lie along the line segment from $(u_x,u_y)$ to
$(v_x,v_y)$.

We want to compute the following: given any terrain $T$ and any
viewpoint $v = (v_x,v_y,v_z)$, find which grid points 
of the terrain are visible to $v$ and which are not.

We assume the terrain is given as a matrix $Z$, stored row by row,
where $Z_{ij}$ is the elevation of the point in row $i$ and column
$j$. The output visibility map is a matrix $V$, stored row by row, in
which $V_{ij}$ is $1$ if the point in row $i$ and column $j$ is
visible, and $0$ otherwise.

For ease of presentation, throughout the rest of the paper we assume
that the grid is square and has size $\sqrt n \cdot \sqrt n$ ; of
course the actual implementations of our algorithms can handle
rectangular grids as well.

\subsection{Related work}  


The standard method for computing viewsheds on grid terrains is the
algorithm \rrr by Franklin and Ray~\cite{franklin:sdh94}. \rrr
determines the visibility of each point in the grid as follows: it
computes the intersections between the horizontal projection of the
line-of-sight and the horizontal and vertical grid lines, and
computes the elevation of the terrain at these intersection points by
linear interpolation.  Since a line of sight intersects $O(\sqrt n)$
grid lines, determining the visibility of a point takes $O(\sqrt n)$
time.  This is considered to be the standard model and \rrr is
considered to produce the ``exact''
viewshed~\cite{izraelevitz:vis}. However, as described by
Franklin and Ray, \rrr runs in $O(n \sqrt n)$~time, which is too slow
in practice, especially for multiple viewshed computations.

A variety of viewshed algorithms have been proposed that optimize \rrr
while approximating in some way the resulting viewshed: Some
algorithms consider only a subset of the $O(n)$ lines-of-sight; others
interpolate the line-of-sight only at a subset of the $O(\sqrt n)$
intersection points with the grid lines; yet others have some other
way of determining in $O(1)$ time whether a point in the grid is
visible.  The optimized viewshed algorithms run in $o(n \sqrt n)$
time, most often $O(n)$.  Examples are \textsc{XDraw} by Franklin and
Ray~\cite{franklin:sdh94}; \textsc{Backtrack} by
Izraelevitz~\cite{izraelevitz:vis}; \rr by Franklin and
Ray~\cite{franklin:sdh94}; and van Kreveld's radial sweep
algorithm~\cite{kreveld:viewshed}---below we describe briefly the 
results which are relevant to this paper.

The algorithm named \rr, proposed by Franklin and
Ray~\cite{franklin:sdh94}, is an optimization of \rrr that runs in
$O(n)$ time.  The idea of \rr is to examine the lines-of-sight
\emph{only} to the $O(\sqrt n)$ grid points on the boundary of the
grid; a grid point that is not on the boundary is considered to be
visible if the nearest point of intersection between a grid line and
one of the examined lines-of-sight is determined to be visible.
Overall \rr is fast and, according to its authors, produces a good
approximation of \rrr that outweighs its loss in
accuracy~\cite{franklin:sdh94}.

The other algorithm, \emph{XDraw}, computes the visibility of the grid
points incrementally in concentric layers around the viewpoint,
starting at the viewpoint and working its way outwards. For a grid
point $v$ in layer $i$, the algorithm computes whether $v$ is visible,
and what is the maximum height above the horizon along the line of
sight to $v$.  To do so, it determines which are the two grid points
$q$ and $r$ in layer $i - 1$ that are nearest to $\overline{pv}$, and
then it estimates the maximum height above the horizon along
$\overline{pv}$ by interpolating between the lines of sight to $q$ and
$r$.
Thus, the visibility of each point is determined in constant time
per point.  
XDraw is faster than R3 and R2, due to the simplicity of the
calculations, but it is also the least accurate~\cite{franklin:sdh94}.
Izraelevitz~\cite{izraelevitz:vis} presented a generalization of
\emph{XDraw} that allows to user to set a parameter $k$, which is the
number of previous layers that are taken into account when computing
the visibility of a grid point. 

\begin{figure}[t]
  \centering{
    \includegraphics[width=4cm]{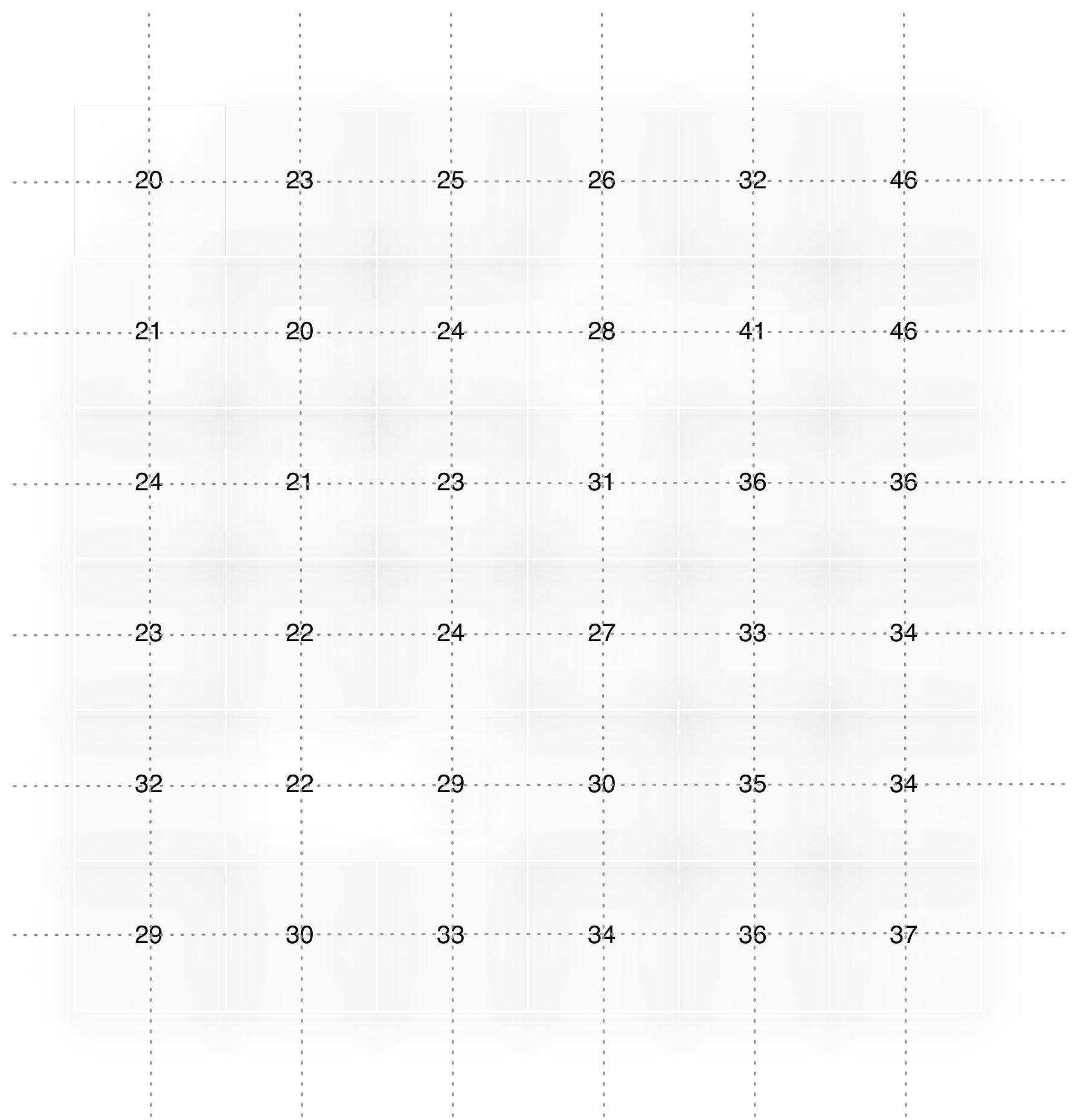}
    \includegraphics[width=4cm]{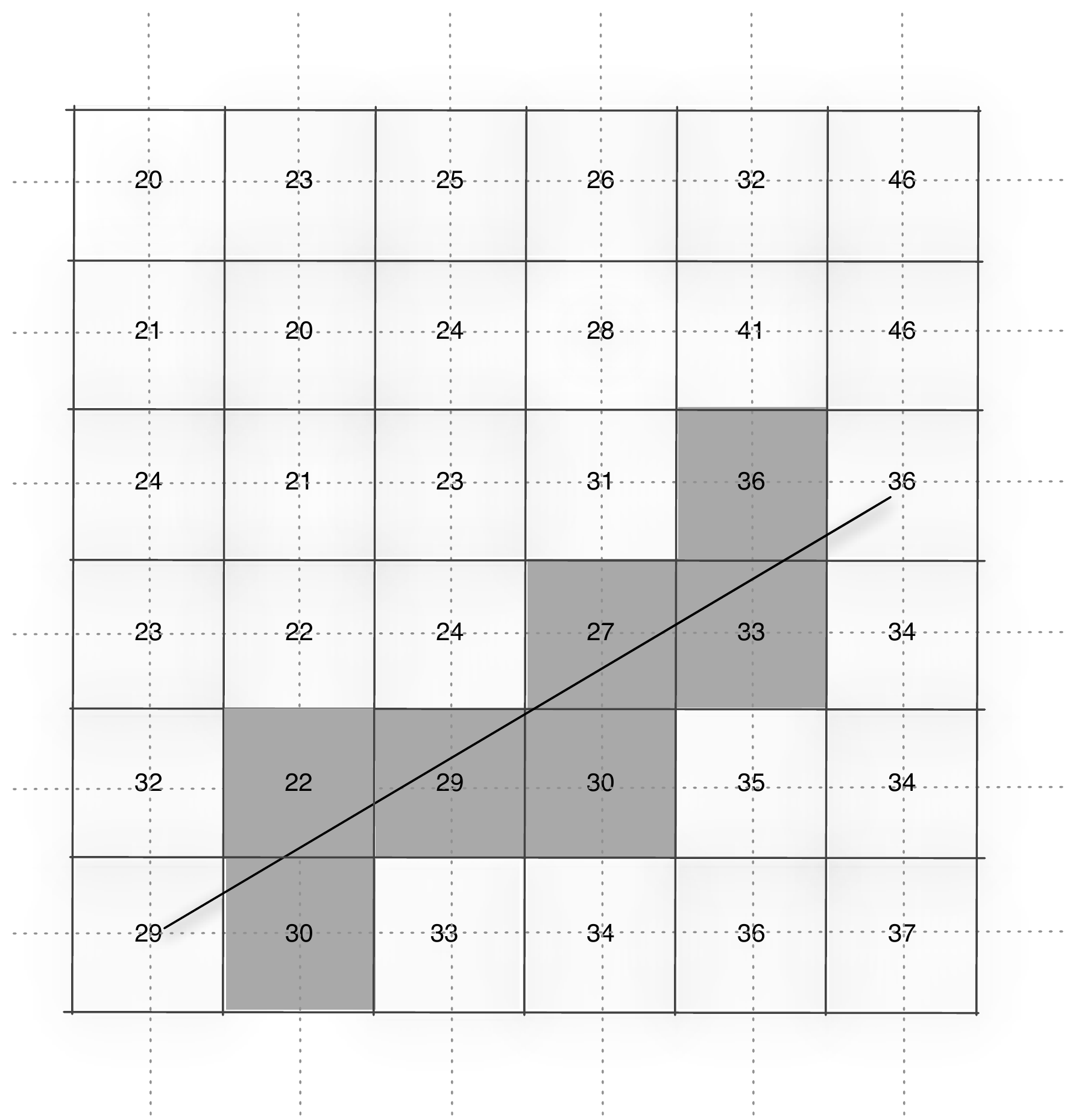}
     \caption{Van Kreveld's model:  (a) Each grid point represents a square
       cell centered at that point.  The elevation angle of all points
       inside a cell is the same as the elevation angle of the center
       of the cell (the grid point).  (b) To determine if two grid points
       are visible,  we need  to consider all cells that intersect the
       line segment beween the points. }
    \label{fig:kreveld-model}
  }
\end{figure}

Van Kreveld described a different approach for computing viewsheds on
grids that could also be seen as an optimization of
\rrr~\cite{kreveld:viewshed}. In his model the terrain is seen as a
tessellation of square cells, where each cell is centered around a
grid point and has the same view angle as the grid point throughout
the cell, that is, the cell appears as a horizontal line segment to
the viewer (Figure~\ref{fig:kreveld-model}). This property allows for
the viewshed to be computed in a radial sweep of the terrain in $O(n
\lg n)$ time. Because cells have constant view angle, they can be
stored in an efficient data structure as the ray rotates around the
viewpoint.  This data structure supports insertions of cells,
deletions of cells, and visibility queries for a point along the ray
in $O(\lg n)$ time per operation, and thus the whole viewshed can be
computed while rotating the ray in $O(n \lg n)$ time.

The viewshed algorithms mentioned so far assume that the computation
fits in memory and are not IO-efficient.  I/O-efficient viewshed
algorithms have been proposed by Magalh\~aes et
al.~\cite{magalhaes:geo}, Ferreira et al.~\cite{ferreira:vis} and in
our previous
work~\cite{havertoma:visibility-journal,ht:vis2,ht:accuracy}; we
discuss these results below.

Haverkort, Toma and Zhuang~\cite{havertoma:visibility-journal}
presented the first IO-efficient viewshed algorithm using Van
Kreveld's model.  Using a technique called \emph{distribution
  sweeping} they turned Van Kreveld's algorithm into an algorithm
running in $O(n \log n)$ time and $O(\sort(n))$~\ios,
cache-obliviously. The authors also presented practical results
showing that their algorithm scales well to large data and outperforms
Van Kreveld's algorithm running in (virtual) memory.

Subsequently, Magalh\~aes et al.~\cite{magalhaes:geo} and Ferreira et
al.~\cite{ferreira:vis} described I/O-efficient versions of Franklin's
R2 algorithm. The first algorithm runs in $O(n \log n)$ time and
$O(\sort(n))$~\ios~\cite{magalhaes:geo}.  As in R2, the idea is to
evaluate lines-of-sight only to the points on the perimeter of the
grid. To do this \io-efficiently, the algorithm first copies all grid
points from the input file row by row, annotating each point $p$ with
the endpoints of the lines of sight whose evaluation requires the
elevation of $p$. Next, all annotated points are sorted by line of
sight. The algorithm then evaluates each line of sight, determining
for each point on a line of sight whether it is visible or not, and
writes the results to a file, in order of computation. As a result,
the file contains the visibility map, ordered by line of sight. The
last step is to sort this file into row-by-row order.

A further improved version of R2 was presented in
~\cite{ferreira:vis}.  Here the idea is to partition the grid in
blocks and run the (in-memory) version of the R2 algorithm modified so
that it bypasses the VMM (virtual memory management) system, and
instead it maintains a data structure of ``active'' blocks that
constitute the block footprint of the algorithm. Whenever the
line-of-sight intersects a block, that block is brought in main
memory. Blocks are evicted using LRU policy. Their algorithm,
\textsc{TiledVS}, consists of three passes: convert the grid to Morton
order, compute visibility using the R2 algorithm, and convert the
output grid from Morton order to row-major order. 
In practice, this
algorithm is much faster than the one 
in~\cite{magalhaes:geo}, achieving on the order of 5,000 seconds on
SRTM dataset of 7.6 billion points (\texttt{SRTM1.region06}, 28.4GiB
using 4 bytes per elevation value) that is, .7$\mu$s per point.
Another advantage of \textsc{TiledVS} is that its first step
can be viewed as a preprocessing step common to all viewpoints and
thus \textsc{TiledVS} computes the viewshed in only two passes over
the grid.


\vspace{.5cm} The IO-efficient algorithms discussed above differ in
how many points they chose to interpolate, how many lines-of-sight
they consider, and how they interpolate the terrain. These choices
affect both the running time and output of the algorithms.
All algorithms described can be considered as approximations of \rrr
and make some assumptions that they exploit to improve efficiency.
\textsc{TiledVS} derives its efficency in part from considering only
$O(\sqrt n)$ LOS's instead of $O(n)$. Van Kreveld's approach exploits
crucially that cells have constant elevation angle across their
azimuth range. Generalizing to linear interpolation is difficult: it
would mean that cells have variable elevation angle across their
azimuth tange, and one would need a kinetic data structure as active
structure to store elevation angles that change in time.

To evaluate viewshed algorithms it is important to consider both
efficiency of running time and accuracy of the computed viewshed. While efficiency
is easy to compare, comparing accuracy is much harder. The
straightforward way to assess accuracy is to compare the computed
viewshed with ground truth data. Ideally one would consider a large
sample of viewpoints, compute the viewshed from each one in turn,
compare it with the \emph{real} viewshed at that point, and aggregate
the differences. Unfortunately, ground truth viewsheds are hard, if
not impossible, to obtain. 


The algorithms mentioned above assume grid terrains.  For
an overview of internal-memory algorithms for visibility computations
on the second most common format of terrain elevation models, the
\emph{triangular irregular network} or TIN, we refer to
~\cite{colesharir:visib,floriani:visdtm,floriani:intervisibility}.
Visibility algorithms on TINs use the concept of a \emph{horizon} or
\emph{silhouette} $\sigma$ of the terrain, which is the upper rim of
the terrain, as it appears to a viewer at $v$. More formally,
$\sigma_T$ is a function from azimuth angles (compass direction) to
elevation angles, such that $\sigma_T(\alpha)$ is the maximum
elevation angle of any point on the intersection of $T$ with the ray
that extends from $v$ in direction $\alpha$.
On a triangulated terrain, the horizon is equivalent to the upper
envelope of the triangle edges of $T$, projected on an infinite vertical cylinder centered on the viewpoint; it has complexity $O(n \cdot
\alpha(n))$, where $\alpha$ is the inverse Ackermann
function~\cite{colesharir:visib}.  Horizons have been used to
solve various visibility-related problems on triangulated polyhedral
terrains. For example, the visibility of all the vertices in a TIN can be
computed in $O(n \alpha(n) \lg n)$ time~\cite{colesharir:visib}. A central
idea in these solutions is that horizons can be merged in time that is linear in their
size, and thus allow for efficient divide-and-conquer algorithms.

\subsection{Our contributions} 

This paper describes IO-efficient algorithms for computing viewsheds
on massive grid terrains in a couple of different models.
 Our first two algorithms work in Van Kreveld's model, and sweep the
 terrain radially by rotating a ray around the viewpoint while
 maintaining the terrain profile along the ray. The difference between
 the two new algorithms is in the preprocessing before the sweep: the
 first algorithm, which we describe in
 Section~\ref{sec:radial-layers}, sorts the grid points in concentric
 bands around the viewpoint; the second algorithm, which we describe
 in Section~\ref{sec:radial-sectors}, sorts the grid points into
 sectors around the viewpoint. Both algorithms run in $O(n \log n)$
 time and $O(\sort(n))$ \ios.

The third algorithm, \cosweep, which we describe in
Section~\ref{sec:conc-sweep}, uses a complementary approach and sweeps
the terrain centrifugally.  
The algorithm is similar to \emph{XDraw}: it grows a region around the
viewpoint, while maintaining the horizon of the terrain within the
region seen so far.  To maintain the horizon efficiently, we represent
it by a grid model itself: we maintain the maximum elevation angle
(the ``height'') of the horizon for a discrete set of regularly spaced
azimuth angle intervals. The horizontal resolution of the horizon
model is chosen to be similar to the horizontal resolution of the
original terrain model, so that we maintain elevation angles for
$\Theta(\sqrt n)$ azimuth angle intervals.  This allows a significant
speed-up as compared to algorithms that process events at $\Theta(n)$
different azimuth angles, or work with horizons of linear complexity.
Also, we note that this gives the algorithm the potential for higher
accuracy than \emph{XDraw}, which represents the horizon up to a given
layer by only as many grid points as there are in that layer---which
can be quite inaccurate close to the viewpoint. Another difference
with \emph{XDraw} is that our algorithm does not proceed layer by
layer, but instead grows the region in a recursive, more \io-efficient
way; this results in a significant speed-up in practice. The
centrifugal sweep algorithm runs in $O(n)$ time and $O(\scan(n))$ \ios
cache-obliviously, and is our fastest algorithm.

Our last two algorithms constitute an improved, IO-efficient version of
Franklin's \rrr algorithm. We distinguish between two models
(Figure~\ref{fig:model}), which we describe in
Section~\ref{sec:algorithms}: In the \emph{gridlines} model we view
the points in the input grid as connected by horizontal and
vertical lines, and visibility is determined by evaluating the
intersections of the line-of-sight with the grid lines using linear
interpolation; this is the model underlying \rrr.

We also consider a slightly different model, the \emph{layers} model,
in which we view the points in the input grid to be connected in
concentric layers around the viewpoint and visibility is determined by
evaluating the intersections of the line-of-sight with these layers
using linear interpolation. The layers model considers only a subset
of the intersections considered by the gridlines model and therefore
the viewshed generated will be larger (more optimistic) than the one
generated with the gridlines model.  Preliminary results
(\cite{ht:vis2}) show that these differences are
practically insignificant. The layers model is faster in practice,
while having practically the same accuracy as the gridlines model.

We describe our last two algorithms, \visiter and \visdac, in
Section~\ref{sec:algorithms}. They are based on
computing and merging horizons in an iterative or divide-and-conquer
approach, respectively. Horizon-based algorithms for visibility
problems have been described by de Floriani and
Magillo~\cite{floriani:visdtm}.  On a triangulated terrain $T$, the
horizon is equivalent to the upper envelope of the triangle edges of
$T$ as projected on a view screen, and has complexity $H(n) = O( n
\cdot \alpha(n))$, where $n$ is the number of vertices in the TIN and
$\alpha()$ is the inverse of the Ackerman
function~\cite{cole:visibility}.  In Section ~\ref{sec:visiter} we
show that we can prove a better bound for our setting: that is, we
prove that the upper envelope of a set $S$ of line segments in the
plane such that the widths of the segments do not differ in length by
more than a factor $d$ has complexity $O(dn)$. From here we show that
the horizon on a grid of $n$ points with linear interpolation has
complexity $O(n)$ in the worst case.

In Section~\ref{sec:results} we describe an experimental analysis and
comparison of our algorithms on datasets up to 28 GB.
All algorithms are scalable to volumes of data that are more than 50
times larger than the main memory.  Our main finding is that, in
practice, horizons are significantly smaller than their theoretical
upper bound, which makes horizon-based algorithms unexpectedly fast.
Our last two algorithms, which compute the most accurate viewshed, turn out to be very fast in practice, although their worst-case bound is inferior.
%
%
%
We conclude in Section~\ref{sec:discussion}.

\section{I/O-Efficient radial sweep}\label{sec:radial-layers}

This section describes our first approach to computing a viewshed. It
is loosely based on Van Kreveld's radial sweep algorithm, which we
present below. 

{\bf The model.} We consider that the terrain is represented by a set
of $n$ points whose projections on~$D$ form a regular grid with
inter-point distance~1. Furthermore, we assume that each grid point $q
= (q_x,q_y,q_z)$ represents a square ``cell'' $D(q)$ on $D$ of size
$1\times 1$, centered on $(q_x,q_y)$. For any given viewpoint $v$, we
treat the terrain above $D(q)$ as if each point of $D(q)$ has
elevation angle $\textrm{ElevAngle}(q)$. This is the interpolation
method used by Van Kreveld~\cite{kreveld:viewshed}.
Determining whether a point $p$ is visible from viewpoint $v$ comes
down to deciding whether there is any other grid point $q$ such that
the square cell $D(q)$ intersects the line segment from $(v_x,v_y)$ to
$(p_x,p_y)$ and $\textrm{ElevAngle}(q) \geq \textrm{ElevAngle}(p)$;
see Figure~\ref{fig:kreveld-model}.

\subsection{Van Kreveld's radial sweep algorithm}
The basic idea of Van Kreveld's algorithm~\cite{kreveld:viewshed} is
to rotate a half-line (ray) around the viewpoint~$v$ and compute the
visibility of each grid point in the terrain when the sweep line
passes over it (see Figure~\ref{fig:kreveld}). For this we maintain a data structure (the
\emph{active} structure) that, at any time in the process, stores the
elevation angles for the cells currently intersected by the sweep line
(the \emph{active cells}). Three types of events happen during the
sweep:

\begin{denseitems}
\item \enterevent events: When a cell starts being intersected by the sweep line, it is inserted in the active structure;
\item \centerevent events: When a the sweep line passes over the grid point $q$ at the center of a cell, the active structure is queried to find out if $q$ is visible.
\item \exitevent events: When a cell stops being intersected by the sweep line, it is deleted from the active structure;
\end{denseitems}

Thus, each cell in the grid has three associated events. Van
Kreveld~\cite{kreveld:viewshed} uses a balanced binary search tree for
the active structure, in which the active cells are stored in order of
increasing distance from the viewpoint. Because the cells are convex,
this is always the same as ordering the active cells in order of
increasing distance from the viewpoint to pothe \emph{grid points}
corresponding to the cells.  With each cell we store its elevation
angle. In addition, each node in the tree is augmented with the
highest elevation angle in the subtree rooted at that node. A query if
a point $q$ is visible is answered by checking if the active structure
contains any cell that lies closer to the viewpoint than $q$ and has
elevation angle at least $\mathrm{ElevAngle}(q)$: if yes, then $q$ is
\emph{not} visible, otherwise it is.
Such a query can be answered in $O(\log n)$ time.
To run the complete algorithm, we first generate and sort the $3n$ events by their azimuth angles (the sweep line directions at which they happen). Then we process the events in order of increasing azimuth angle. The whole algorithm runs in $O(n \log n)$ time.

\begin{figure}[t]
  \centering{
    \includegraphics[width=5cm]{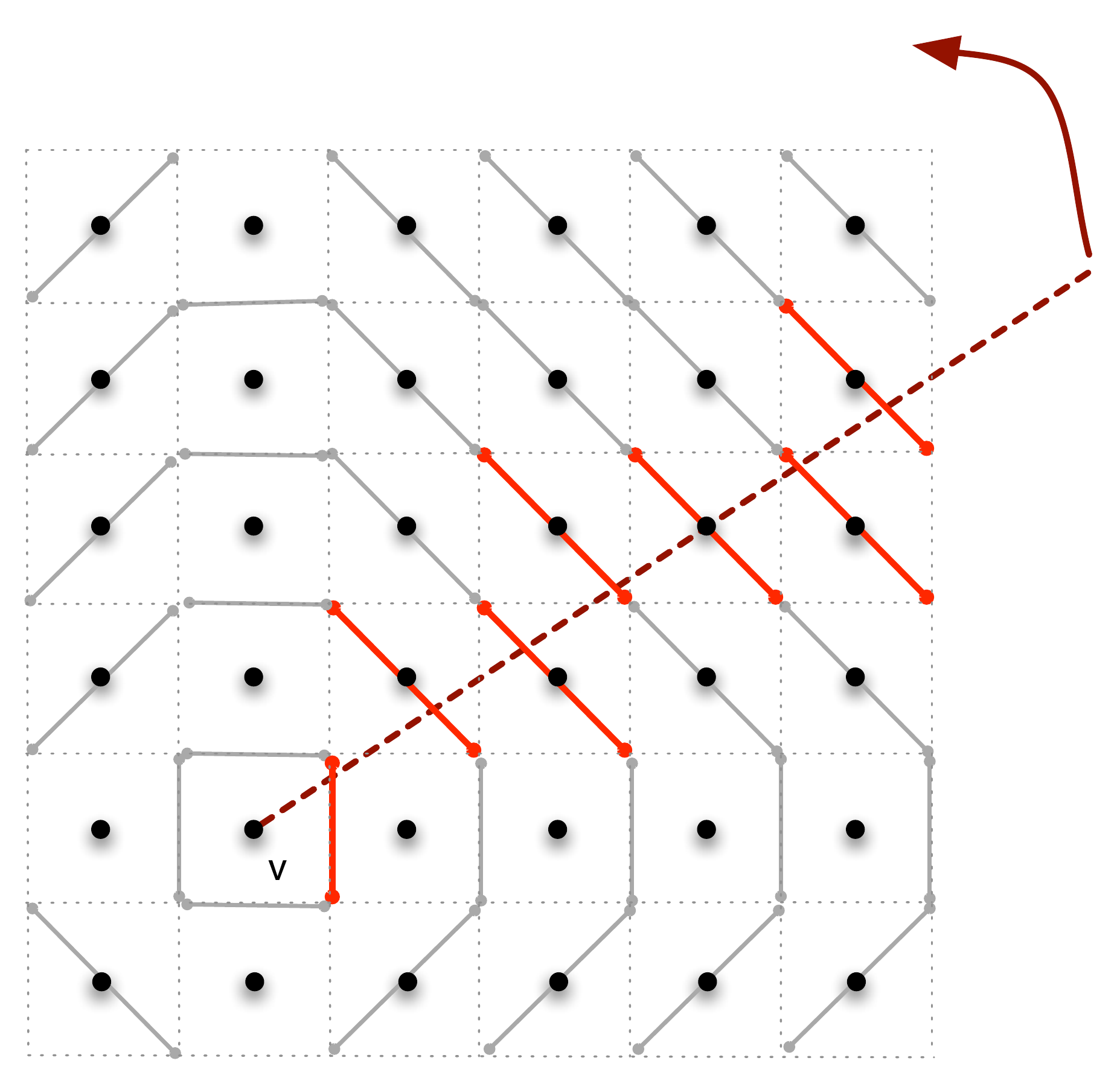}
    \caption{Van Kreveld algorithm. Cells will be present in the
      active structure as long as the sweep ray intersects the
      indicated diagonal.}
    \label{fig:kreveld}
  }
\end{figure}

In our previous work we adapted Van Kreveld's algorithm to make it
I/O-efficient~\cite{havertoma:visibility-journal}. The first step was
still to generate and sort the events. For each event we stored its
location in the plane and its elevation angle. Using four bytes per
coordinate, this resulted in an event stream of $36n$ bytes. For large
$n$, this is a significant bottleneck.

\subsection{A new I/O-efficient radial sweep algorithm}\label{sec:iorad2}

The main idea of our new radial sweep algorithm is therefore to avoid
generating and sorting a fully specified event stream. The purpose of
the event stream was to supply the azimuth angle and the elevation
angle of the events in order. Note, however, that the azimuth angle of
the events only depends on how the sweep progresses over the grid, but
not on the elevation values stored in the input file. Only the
elevation angles have to be derived from the input file.

Our ideas for making the sweep I/O-efficient are now the following. We can
compute the azimuth angles of the events on the fly, without accessing
the input file, instead of computing all events in advance. Only when
processing an \enterevent event corresponding to a grid point $q$, the
elevation of $q$ needs to be retrieved in order to insert $\langle
\mathrm{Dist}(q), \mathrm{ElevAngle}(q) \rangle$ into the active
structure---for \centerevent events the elevation angle can then be found
in the active structure and for \exitevent events the elevation angle
is not needed. To allow efficient retrieval of elevations for
\enterevent events, we pre-sort the elevation grid into lists of elevation
values, stored in the order of the \enterevent events that require
them. Thus we can retrieve all elevation values in $O(\scan(n))$
\ios during the sweep. Sorting the \emph{complete} elevation grid
into a \emph{single} list would be relatively expensive (it would
require several sorting passes); we avoid that by dividing the grid into
concentric bands around the viewpoint, making one list of elevation values for each band.
As long as the number of bands in small enough so that we can keep a
read buffer of size $\Theta(B)$ for each band in memory during the
sweep, we will still be able to retrieve all elevation values during
the sweep in $\Theta(\scan(n))$ \ios.

\vspace{\baselineskip}
{\bf Notation.}
For ease of description, assume that the viewpoint $v$ is in the
center of the grid at coordinates $(0,0,0)$ and the grid has size $(2m
+ 1) \times (2m + 1)$, where $m = (\sqrt n - 1)/2$.  The elevations of
the grid points are given in a two-dimensional matrix $Z$ that is
ordered row by row, with rows numbered from $-m$ to $m$
from north to south and columns numbered from $-m$ to $m$
from west to east. By $p(i,j)$ we denote the grid point $q =
(q_x,q_y,q_z)$ in row $i$ and column $j$ with coordinates $q_x = j$,
$q_y = -i$ and $q_z = Z_{ij}$; by \emph{cell $(i,j)$} we denote the
square $D(p(i,j))$. Let $\enterevent (i,j)$ denote the azimuth angle
of the \enterevent event of cell $(i,j)$.

\vspace{\baselineskip}
{\bf Description of the algorithm.}
We now describe our algorithm in detail. Let \emph{layer} $l$ of the
grid denote the set of grid points whose $L_\infty$-distance from the
viewpoint, measured in the horizontal plane, is $l$.
We divide the grid in concentric bands of width $w$ around the
viewpoint.  Band $k$ (denoted $B_k$), with $k > 0$, contains all grid
points of layers $(k-1) w + 1$ up to $k w$, inclusive; so $p(i,j)$
would be found in band $\lceil\max(|i|,|j|) / w\rceil$ (see Figure~\ref{fig:layered-radial-sweep}). We choose $w = \Theta(\sqrt M)$; more precisely, $w$ is the largest power of two such
that the elevation and visibility values of a square tile of
$(w+1) \cdot (w+1)$ points fit in one third of the memory.

\begin{figure}[t]
\centering
\includegraphics[height=2.5in]{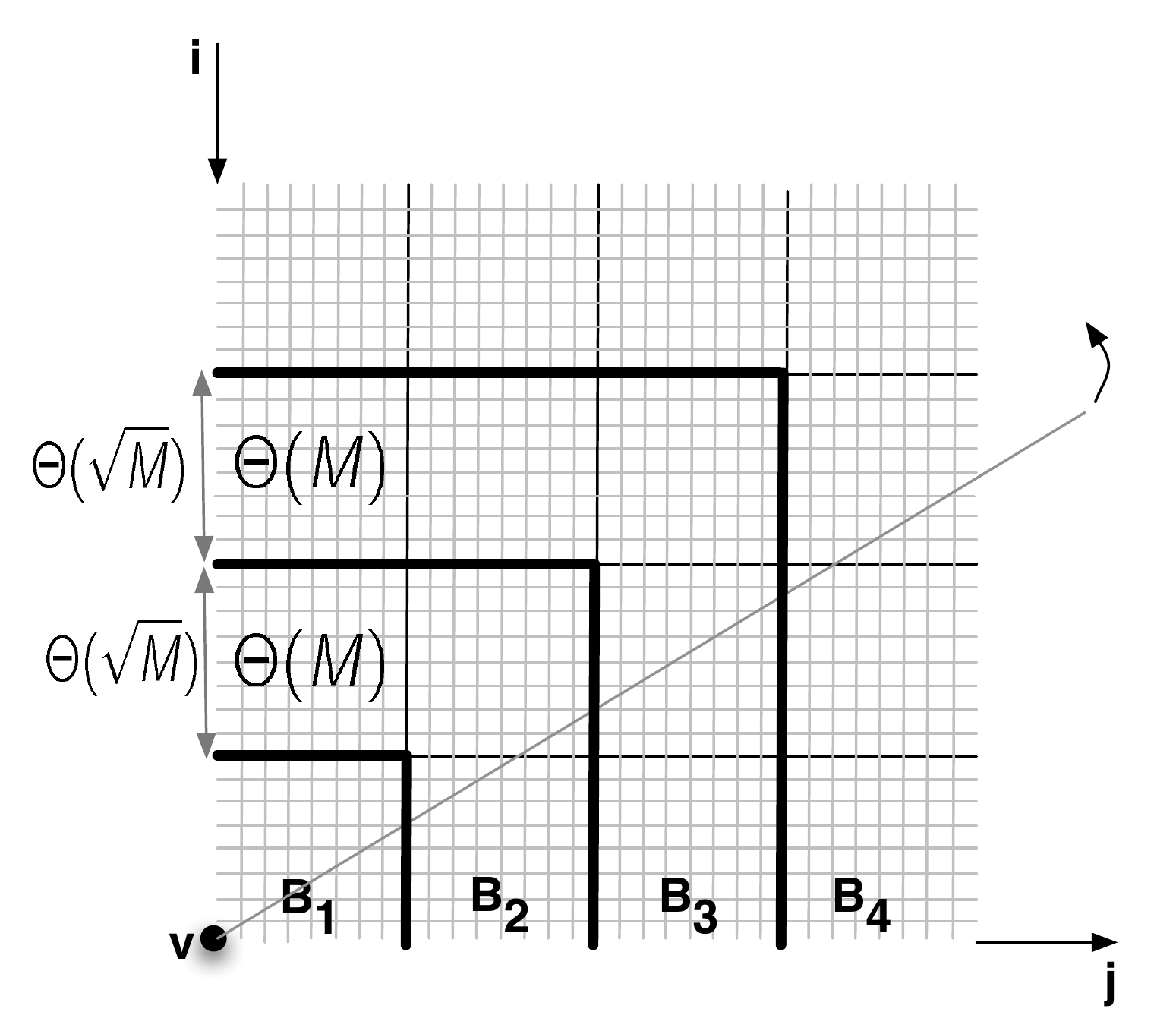}
\caption{A layered radial sweep.}
\label{fig:layered-radial-sweep}
\end{figure}

Our algorithm proceeds in three phases.  The first phase is to
generate, for each band $B_k$, a list $E_k$ containing the elevations
of all points $p(i,j)$ in the band, ordered by increasing
$\enterevent(i,j)$ values (recall that $\enterevent(i,j)$ denotes
the azimuth angle of the enter event of the cell $(i,j)$). Points
$p(i,j)$ with the same $\enterevent(i,j)$ value are ordered
secondarily by increasing distance $\mathrm{Dist}(i,j)$ from the
viewpoint. The algorithm that builds the lists $E_k$ is given below.
The basic idea is to read the grid points from the elevation grid
going in counter-clockwise order around the viewpoint. This is
achieved by maintaining a priority queue with points just in front of
the sweep line; the priority queue is organised by the azimuth angles
of the enter events corresponding to the points to be read.
The queue is initialised with all points of $B_k$ that lie straight
right of the viewpoint (Note: such a point $(0,j)$ will have its
$\enterevent(0,j)$ given by the south-west corner of its cell at
$(1/2, j-1/2)$ which corresponds to an angle in the fourth quadrant
$(3\pi/2, 2\pi)$; we subtract $2\pi$ to bring it to $(-\pi/2, 0)$;
this guarantees that points straight right of the viewpoint are first
in radial order).
Then we extract points from the queue one by
one in order of increasing
$\enterevent(i,j)$; when we extract a point, we read its elevation
from the elevation grid, write the elevation value to $E_k$, and
insert the next point from the same layer in the priority queue (this
is the point above, to the left, below, or to the right, depending on
which octant the current point is in). In this way, from neighbor to
neighbor, all points are eventually reached. Below we describe the
algorithm only for the first quadrant
(Figure~\ref{fig:layered-radial-sweep}); the others are handled
similarly.

\addvspace{.5\baselineskip}
\noindent
\textbf{Algorithm} \textsc{BuildBands}:\\[.25\baselineskip]
\textbf{for }$k \leftarrow 1$ \textbf{to} $\lceil m /w \rceil$\\
\textbf{do }initialise empty list $E_k$ and priority queue $Q$\\
\hphantom{\textbf{do }}\textbf{for }$j \leftarrow (k-1)\cdot w + 1$ \textbf{to} $k \cdot w$\\
\hphantom{\textbf{do }}\textbf{do }insert $\langle \enterevent (0,j)-2\pi, 0,j \rangle$ into $Q$\\
\hphantom{\textbf{do }}\textbf{while} $E_k$ is not complete\\
\hphantom{\textbf{do }}\textbf{do }$\langle \alpha, i, j \rangle \leftarrow Q$.deleteMin()\\
\hphantom{\textbf{do do }}read $Z_{ij}$ from the grid and write it to $E_k$\\
\hphantom{\textbf{do do }}\textbf{if }$-i < j$\hfill\textit{(next cell is north)}\\
\hphantom{\textbf{do do }}\textbf{then }insert $\langle \enterevent (i-1,j), i-1,j \rangle$ into $Q$\\
\hphantom{\textbf{do do }}\textbf{else}\hfill\textit{(next cell is west)}\\
\hphantom{\textbf{do do then }}insert $\langle \enterevent (i,j-1), i,j-1 \rangle$ into $Q$\\
\hphantom{\textbf{do }}clear $Q$

\addvspace{.5\baselineskip} After constructing the lists $E_k$, the
second phase of the algorithm starts: computing which points are
visible. To do this we perform a radial sweep of all events in azimuth
order. Again, we generate the events on the fly with the help of a
priority queue, using only the horizontal location of the grid
points. We use a priority queue to hold events in front of the sweep
line, and an active structure to store the cells that currently
intersect the sweep line, sorted by increasing distance from the
viewpoint (as in Van Kreveld's algorithm). The algorithm starts by
inserting all \enterevent events of the points straight to the right
of the viewpoint into the priority queue. When the next event in the
priority queue is an \enterevent event for cell $(i,j)$, the algorithm
inserts the corresponding \centerevent and \exitevent events in the
queue, as well as the \enterevent event of the next cell in the same
layer. In addition, it reads the elevation $Z_{ij}$ of $p(i,j)$ from
the list of elevation values $E_k$ of the band $B_k$ that contains
$p(i,j)$, and it inserts the cell $(i,j)$ in the active structure with
key $\textrm{Dist}(i,j)$. When the next event in the priority queue is
a $\centerevent$ event for cell $(i,j)$, the algorithm queries the
active structure for the visibility of the point with key
$\textrm{Dist}(i,j)$. When the next event in the priority queue is an
$\exitevent$ event for cell $(i,j)$, the algorithm deletes the element
with key $\textrm{Dist}(i,j)$ from the active structure.

\addvspace{.5\baselineskip}
\noindent
\textbf{Algorithm} \textsc{ComputeVisibility}:\\[.25\baselineskip]
Initialise empty active structure $A$ and priority queue $Q$\\
\textbf{for }$j \leftarrow 1$ \textbf{to} $m$\\
\textbf{do }insert $\langle \enterevent (0,j)-2\pi, \enterevent, 0,j \rangle$ into $Q$\\
\textbf{for }$k \leftarrow 1$ \textbf{to} $\lceil m /w \rceil$\\
\textbf{do }set read pointer of $E_k$ at the beginning\\
\hphantom{\textbf{do }}initialise empty list $V_k$\\
\textbf{while} not all visibility values have been computed\\
\textbf{do }$\langle \alpha, \mathit{type}, i, j \rangle \leftarrow Q$.deleteMin()\\
\hphantom{\textbf{do }}\textbf{if }$\mathit{type} = \enterevent$\\
\hphantom{\textbf{do }}\textbf{then }insert $\langle \centerevent(i,j), \centerevent, i,j\rangle$ in $Q$\\
\hphantom{\textbf{do then }}insert $\langle \exitevent(i,j), \exitevent, i,j\rangle$ into $Q$\\
\hphantom{\textbf{do then }}\textbf{if }$|i| < j$ or $i = j > 0$\quad\quad\textit{(next cell is north)}\\
\hphantom{\textbf{do then }}\textbf{then }insert $\langle \enterevent (i-1,j), \enterevent, i-1,j \rangle$ in $Q$\\
\hphantom{\textbf{do then }}[... \textit{similar for west, south, and east} ...]\\
\hphantom{\textbf{do then }}compute band number $k \leftarrow \lceil\max(|i|,|j|) / w\rceil$\\
\hphantom{\textbf{do then }}$z \leftarrow$ the next unread value from $E_k$\\
\hphantom{\textbf{do then }}$\beta \leftarrow \arctan(z/\mathrm{Dist}(i,j))$\quad\quad($= \mathrm{ElevAngle}(p(i,j))$)\\
\hphantom{\textbf{do then }}insert $\langle \mathrm{Dist}(i,j), \beta\rangle$ into $A$\\
\hphantom{\textbf{do }}\textbf{else if }$\mathit{type} = \centerevent$\\
\hphantom{\textbf{do }}\textbf{then }compute band number $k \leftarrow \lceil\max(|i|,|j|) / w\rceil$\\
\hphantom{\textbf{do then }}query $A$ if element with key $\mathrm{Dist}(i,j)$ is visible;\\
\hphantom{\textbf{do then }}if yes, write 1 to $V_k$, otherwise write 0 to $V_k$\\
\hphantom{\textbf{do }}\textbf{else }($\mathit{type} = \exitevent$)\\
\hphantom{\textbf{do then }}delete element with key $\mathrm{Dist}(i,j)$ from $A$

\addvspace{.5\baselineskip}

The crux of the \textsc{ComputeVisibility} algorithm above is the
following: when it needs to read $Z_{ij}$, it simply takes the next
unread value from its band $E_k$. This is correct, because within each
band $B_k$, the above algorithm requires the $Z_{ij}$ values in the
order of the corresponding $\enterevent$ events, and this is exactly
the order in which these values were put in $E_k$ by algorithm
\textsc{BuildBands}. The output of the second phase is a number of
lists $V_k$ with visibility values: one list for each band, in order
of the azimuth angle of the grid points.

The third phase of the algorithm sorts the lists 
$V_k$ into one visibility map. To do so we run an algorithm that is
more or less the reverse of algorithm \textsc{BuildBands}: we only
need to swap the roles of reading and writing, and use azimuth values
for \centerevent events instead of \enterevent events.

\vspace{\baselineskip}
{\bf Efficiency analysis.}
We will now argue that the above algorithm computes a visibility map
in $O(n \log n)$ time and $O(\scan(n))$ \ios under the assumption
that the input grid is square, and we have $M \geq c_1 \sqrt n$ and $M
\geq c_2 B^2$ for sufficiently large constants $c_1$ and $c_2$.

We start with the first phase: \textsc{BuildBands}.  Consider the part
of band $B_1$ which lies in the first quadrant. This part consists of
all points $p(i,j)$ such that $0 \leq -i \leq w$ and $0 \leq j \leq w$
(except the viewpoint itself). It is a tile of size $(w+1) \cdot
(w+1)$, which fits in one third of the main memory by definition of
$w$. As the algorithm iterates through the points of $B_1$, it
accesses their elevations, loading blocks from disk, until eventually
the entire tile is in main memory, after which there are no subsequent
\io-operations on the input grid. The number of \ios to access the
tile is $O(w + w^2/B) = O(\sqrt{M} + M^2/B)$. By the assumption that
$M = \Omega(B^2)$, this is $O(M^2/B)=O(|B_1|/B)$, where $|B_1|$
denotes the number of grid points in $B_1$.
In fact any band $B_k$ with $k \geq 1$ can be subdivided into $8k-4$
tiles of size at most $(w+1) \cdot (w+1)$, such that for any band, the
sweep line will intersect at most two such tiles at any time (see
Figure~\ref{fig:layered-radial-sweep}). Since a tile fits in at most one
third of the memory, two tiles fit in memory together. Therefore the
algorithm can process each band by reading tiles one by one, without
ever reading the same tile twice. Thus each band $B_k$ is read in
$O(\scan(|B_k|))$ \ios, and algorithm \textsc{BuildBands} needs
$O(\scan(n))$ \ios in total to read the input. The output lists
$E_k$ are written sequentially, taking $O(\scan(n))$ \ios as
well. It remains to discuss the operation of the priority queue. Note
that at any time the priority queue stores one cell from each layer,
and therefore it has size $m < \frac12 \sqrt n$; by assumption this is
at most $\frac12 M / c_1$. Hence, for a sufficiently large value of
$c_1$, the priority queue fits in memory together with the two tiles
from the input file mentioned above (which each take at most one third
of the memory).
Thus the operation of the priority queue takes no \io, but it will
take $O(\log n)$ CPU-time per operation, and thus, $O(n \log n)$ time
in total.

The second phase, \textsc{ComputeVisibility}, reads and writes each
list $E_k$ and $V_k$ in a strictly sequential manner.  There are
$O(m/w) = O(\sqrt{n/M})$ bands. Under the assumption $M \geq c_1 \sqrt
n$ and $M \geq c_2 B^2$, this is only $O(\sqrt{n/M}) = O(\sqrt{n}/B) =
O(M/B)$. This implies that, when $c_1$ and $c_2$ are sufficiently high
constants, one block from each list $E_k$ or $V_k$ can reside in
memory as a read or write buffer during the sweeping. Thus all lists
$E_k$ and $V_k$ can be read and written in parallel in $O(\scan(n))$
\ios in total. The priority queue and the active structure have size
$O(\sqrt n)$ and therefore fit in memory by the arguments given above,
so
the second phase needs $O(n \log n)$ time and $O(\scan(n))$ \ios in total.

The third phase, sorting the output lists into a visibility grid, also takes $O(n \log n)$ time and $O(\scan(n))$ \ios: the analysis is the same as for the first phase.
%
%
Note that in practice, the number of visible points is often very small compared to the size of the grid. In that case it may be better to change the algorithm
\textsc{ComputeVisibility}
as follows: instead of writing the visibility values of \emph{all}
grid points to separate lists for each band and sorting these into a
grid, we record only the \emph{visible} grid points with their grid
coordinates, write them to a single list~$V$, sort this list, and
produce a visibility map from the sorted output.

\subsection{An algorithm for very large inputs}
The above algorithm computes a visibility map in 
$\Theta(\scan(n))$~\ios under the assumption that $M \geq c_1 \sqrt
n$, and $M \geq c_2 B^2$ for sufficiently large constants $c_1$ and
$c_2$.  Note that $\sort(n) = \Theta(\scan(n))$ under these
assumptions.  The idea of a layered radial sweep can be extended to a
recursive algorithm that runs in $O(\sort(n))$ \ios for any $n$,
without both these assumptions.

The idea is the following: we divide the problem into $\Theta(M/B)$
bands, scan the input to distribute the grid points into separate
lists for each band, then compute visibility recursively in each band,
and merge the results. More precisely, for each band we will compute a
list of ``locally'' visible points and a ``local'' horizon: these are
the points and the horizon that would be visible in absence of the
terrain between the viewpoint and the band.  The list of visible
points is stored in azimuth order around the viewpoint.  The horizon
is a step function whose complexity is linear in the number of points
of the terrain; it is also stored as a list of points in azimuth order
around the viewpoint.  

Now we can merge two adjacent bands as follows. Let $V_1$ and $H_1$ be
the list of visible points and the horizon of the inner band, and let
$V_2$ and $H_2$ be the list of visible points and the horizon of the
outer band. The merge proceeds as follows. We scan these four lists in
parallel, in azimuth order, and output two lists in azimuth
order. First, a list of visible points containing all points of $V_1$,
and all points $V_2$ that are visible above $H_1$. Second, the merged
horizon: the upper envelope of $H_1$ and $H_2$.
This correctly computes visibility because a point is visible if and
only if is visible in its band, and is not occluded by any of the
bands that are closer to $v$.  

The idea of the merge step can be extended to merge $M/B$ bands,
resulting in an algorithm that runs in $O(\sort(n))$ \ios.
To see why, observe that there are $\Theta(M/B)$ levels of recursion,
and the base-case runs in linear time. Each band and its horizon have
size $O(n/(M/B))$. The merging can be performed in linear time because
it involves scanning of $\Theta(M/B)$ lists of size $\Theta(n/(M/B))$;
a block from each list fits in memory and the total size of all lists
is $\Theta(n)$. It remains to show that a horizon can be computed
in linear time in a base-case band of width $\Theta(\sqrt M)$. To see
this, we note that a band consists of tiles of size $\sqrt M \cdot
\sqrt M$, and the horizon can be computed tile by tile. The details
are similar to ones already discussed, and we omit them.

Overall, we get a divide-and-conquer algorithm that can compute
visibility in $\Theta(\sort(n))$ \ios for any $n$, assuming $M \geq
c_2 B^2$.  Because another algorithm with a theoretical I/O-efficiency
of $O(\sort(n))$ was already known from our previous
work~\cite{havertoma:visibility-journal}, this ``new''
divide-and-conquer algorithm is not particularly interesting.
In practice such a recursive algorithm would probably never be needed:
it would only be useful when $n$ would be at least as big as
$(M/c_1)^2$.

\section{A radial sweep in sectors}\label{sec:radial-sectors}

This section describes our second algorithm for computing the
visibility map of a point $v$.  It does not achieve better asymptotic
bounds on running time and \ios than the algorithm from the previous
section, but, as we will see in Section~\ref{sec:results}, it is
faster.  Like the algorithm from Section~\ref{sec:radial-layers}, our
second algorithm sweeps the terrain radially around the viewpoint. As
before, the azimuth angles of the events are computed on the fly using
a priority queue. Elevation values of grid points are only needed when
their \enterevent events are processed. To make access to elevation
values efficient, we first divide the elevation grid into sectors of
$\Theta(M)$ grid points each---this is the main difference with the
algorithm from the previous section, which divided the elevation grid
into concentric bands.

The algorithm proceeds in three phases. First, for any pair of azimuth
angles $\alpha,\beta$, let $S(\alpha,\beta)$ be the set of grid points whose
corresponding \enterevent events have azimuth angle at least $\alpha$ and
less than $\beta$. The first phase of our algorithm starts by computing a
set of azimuth angles $\alpha_0 < ... < \alpha_s$, where $\alpha_0 =
0$ and $\alpha_s = 2\pi$, such that for any $1 \leq k \leq s$ we have
that the coordinates and elevation values of
$S(\alpha_{k-1},\alpha_k)$ fit in one third of the main memory. Note
that this can be done without accessing the elevation grid: the
algorithm only needs to know the size of the grid and the location of
the viewpoint in order to be able to divide the full grid into
memory-size sectors. We then scan the elevation grid
and distribute the grid points based on their \enterevent azimuth
angle into lists: one list $E_k$ for each sector
$S(\alpha_{k-1},\alpha_k)$. (Cells straight right of the viewpoint
need to be entered at the beginning of the sweep and are additionally
put in $E_1$).

In the second phase we do the radial sweep as before, sector by
sector, with two modifications: (i) whenever we enter a new sector
$S(\alpha_{k-1},\alpha_k)$, we load the complete list $E_k$ into
memory and sort it by the azimuth angle of the \enterevent events;
(ii) we do not keep a list of visibility values per sector, but
instead we write the row and column coordinates of the points that are
found to be visible to a single list $L$.

\begin{figure}[t]
\centering
\includegraphics[height=2.5in]{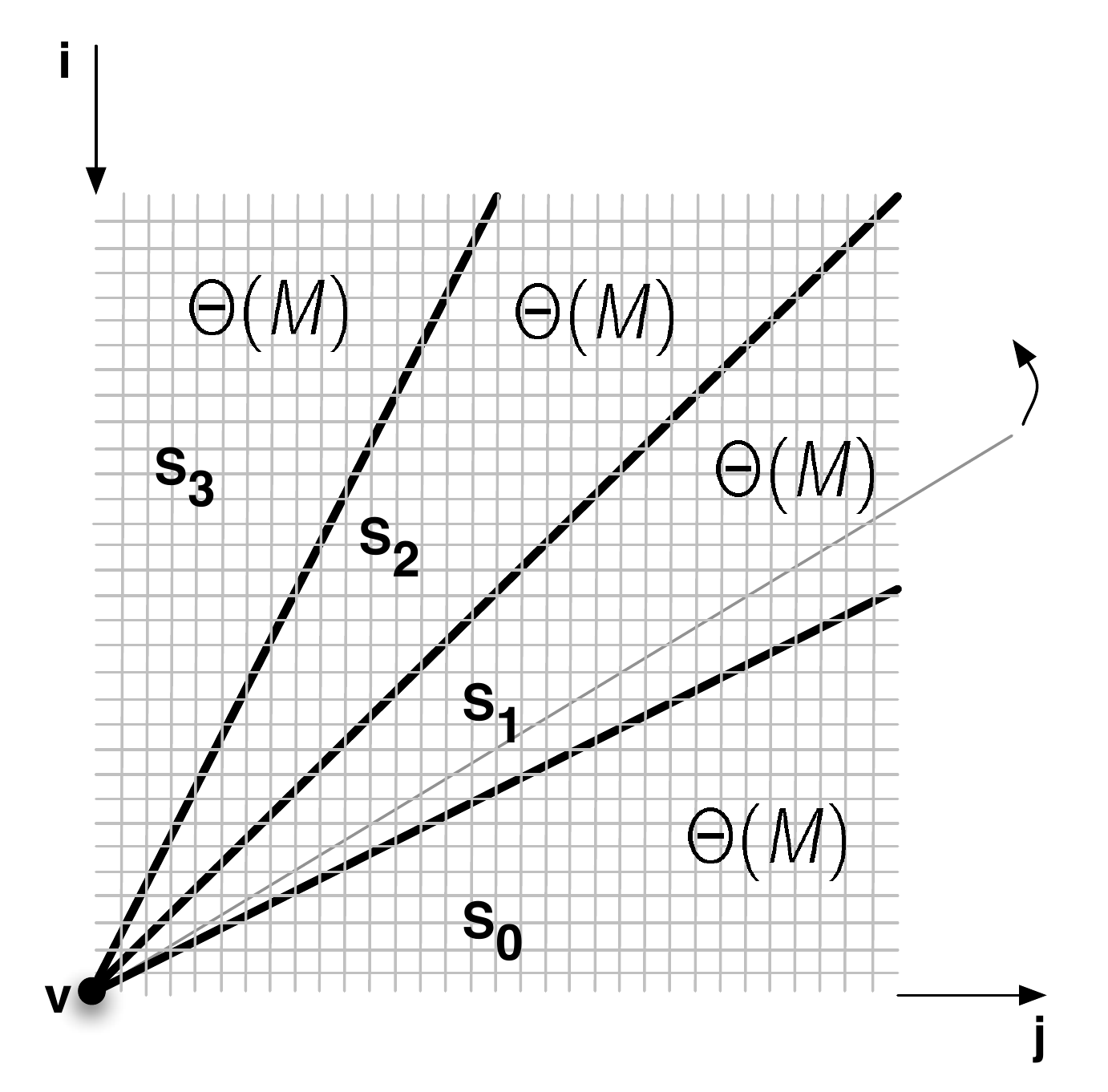}
\caption{A radial sweep in sectors.}
\label{fig:sector-radial-sweep}
\end{figure}

Finally, in the third phase we sort $L$ and scan it to produce a
visibility map of the full grid.  Thus the full algorithm is as
follows:
\newpage
\addvspace{.5\baselineskip}
\noindent
\textbf{Algorithm} \textsc{SectoredSweep}:\\[.25\baselineskip]
\textit{First phase---distribution:}\\
Compute sector boundaries
$\alpha_0,...,\alpha_s$ analytically such that each sector
$S(\alpha_{k-1},\alpha_k)$ fits in one third of the memory.\\
\textbf{for }$k \leftarrow 0$ \textbf{to} $s$\\
\textbf{do }initialise empty list $E_k$\\
\textbf{for }all points $p(i,j)$ in row-by-row order (except $v$)\\
\textbf{do }compute $k$ s.t. $\alpha_{k-1} \leq \enterevent(i,j) < \alpha_k$\\
\hphantom{\textbf{do }}read $Z_{ij}$ and write $\langle i,j,Z_{ij}\rangle$ to $E_k$\\
\textbf{for }all points $p(i,j)$ from $v$ (excl.) straight to the right\\
\textbf{do }read $Z_{ij}$ and write $\langle i,j,Z_{ij}\rangle$ to $E_0$\\[.25\baselineskip]
\textit{Second phase---sweep:}\\
initialise empty active structure $A$ and priority queue $Q$\\
initialise empty output list $L$\\
\textbf{for }$j \leftarrow 1$ \textbf{to} $m$\\
\textbf{do }insert $\langle \enterevent (0,j)-2\pi, \enterevent, 0,j \rangle$ into $Q$\\
$k \leftarrow 1$; load $E_1$ in memory and sort it by $\enterevent(i,j)$\\
\textbf{while} not all visibility values have been computed\\
\textbf{do }$\langle \alpha, \mathit{type}, i, j \rangle \leftarrow Q$.deleteMin()\\
\hphantom{\textbf{do }}\textbf{if }$\mathit{type} = \enterevent$\\
\hphantom{\textbf{do }}\textbf{then }\textbf{if }$E_k$ contains no more unread elements\\
\hphantom{\textbf{do then }}\textbf{then }delete $E_k$; $k \leftarrow k + 1$; load $E_k$ in memory\\
\hphantom{\textbf{do then then }}sort $E_k$ and set read pointer at beginning\\
\hphantom{\textbf{do then }}$z \leftarrow$ read the next unread value from $E_k$\quad\quad (=$Z_{ij}$)  \\
\hphantom{\textbf{do then }}$\beta \leftarrow \arctan(z/\mathrm{Dist}(i,j))$\quad\quad($= \mathrm{ElevAngle}(p(i,j))$)\\
\hphantom{\textbf{do then }}insert $\langle \mathrm{Dist}(i,j), \beta\rangle$ into $A$\\
\hphantom{\textbf{do then }}insert $\langle \centerevent(i,j), \centerevent, i,j\rangle$ in $Q$\\
\hphantom{\textbf{do then }}insert $\langle \exitevent(i,j), \exitevent, i,j\rangle$ into $Q$\\
\hphantom{\textbf{do then }}\textbf{if }$|i| < j$ or $i = j > 0$\quad\quad\textit{(next cell is north)}\\
\hphantom{\textbf{do then }}\textbf{then }insert $\langle \enterevent (i-1,j), \enterevent, i-1,j \rangle$ in $Q$\\
\hphantom{\textbf{do then }}[... \textit{similar for west, south, and east} ...]\\
\hphantom{\textbf{do }}\textbf{else if }$\mathit{type} = \centerevent$\\
\hphantom{\textbf{do }}\textbf{then }query $A$ if element with key $\mathrm{Dist}(i,j)$ is visible;\\
\hphantom{\textbf{do then }}if yes, write $\langle i,j\rangle$ to $L$\\
\hphantom{\textbf{do }}\textbf{else }($\mathit{type} = \exitevent$)\\
\hphantom{\textbf{do then }}delete element with key $\mathrm{Dist}(i,j)$ from $A$\\[.25\baselineskip]
\textit{Third phase---produce visibility map:} \\
Sort $L$ lexicographically by row, column\\
 Set read pointer of $L$ at the beginning\\
 \textbf{for }all points $p(i,j)$ in row-by-row order\\
 \textbf{do if }next element of $L$ is $(i,j)$\\
 \hphantom{\textbf{do }}\textbf{then }$V_{ij} \leftarrow 1$; advance read pointer of $L$\\
 \hphantom{\textbf{do }}\textbf{else }$V_{ij} \leftarrow 0$\\

\vspace{\baselineskip}
{\bf Efficiency analysis.}
We will now briefly argue that the above algorithm computes a
visibility map of the first quadrant in $O(n \log n)$ time and
$O(\scan(n) + \sort(t))$ \ios, where $t$ is the number of visible
grid points, under the assumption that the input grid is square and
$M^2/B \geq c n$ for a sufficiently large constant $c$.

The first phase of the algorithm reads the elevation grid once
and writes elevation values to $O(n/M) = O(M/B)$ sector
lists. Therefore we can keep, for each sector, one block of size
$\Theta(B)$ in memory as a write buffer, and thus the first phase
produces the sector lists in $O(\scan(n))$ \ios. The running time of
the first phase is $\Theta(n)$.

During the second phase, we read the sector lists one by one, in
$O(\scan(n))$ \ios in total. The priority queue and the active
structure can be maintained in memory by the arguments given in the
previous section. Creating and sorting $L$ takes $O(\sort(t))$ \ios,
after which it is scanned to produce a visibility map.

Thus the  algorithm runs in $O(n \log n)$ time and $O(\scan(n)
+ \sort(t))$ \ios.

\vspace{\baselineskip}
{\bf An algorithm for very large inputs.}
When the assumption $M^2/B \geq c n$ does not hold, a radial sweep
based on distribution into sectors is still possible: one can use the
recursive distribution sweep algorithm from our previous
work~\cite{havertoma:visibility-journal} and apply the ideas described
above to reduce the size of the event stream. The result is an
algorithm that runs in $O(n \log n)$ time and $O(\sort(n))$ \ios.
We sketch the main ideas below. 

First we note that, for large $n$, the ``diagonal'' of the grid is
larger than $M$ and does not fit in a sector. The splitter values
for sectors can still be computed without any \io because they
depend solely on the position $(i,j)$ of the points wrt $v$, and not
on their elevation. For example, one could do a pass through the
points in $\enterevent(i,j)$ order using an I/O-efficient priority
queue, in the same way as during the sweep, but without accessing
the elevation. Using a counter we can keep track of the number of
points processed, and output every $\Theta(M)$-th one as a splitter.

Given the splitters, we can proceed recursively: first distribute the
grid into $\Theta(M/B)$ sectors, and then distribute each sector
recursively until each sector has size $\Theta(M)$. This takes
$\Theta(\log_{M/B} n)$ passes over the grid. Thus, distribution into
$\Theta(M)$-sized sectors can be performed in $\Theta(\sort(n))$
\ios. If $n > M^2$, the active structure does not fit in memory and
the sweep of the sectors with a common active structure does not work,
even though each sector is $\Theta(M)$. We need to refine the
distribution to process carefully long cells that span more than one
sector, so that we can process each sector individually. This can be
done I/O-efficiently in $\Theta(\sort(n))$ \ios and we refer to our
previous algorithm for details~\cite{havertoma:visibility-journal}.

\section{A centrifugal sweep algorithm}\label{sec:conc-sweep}
In this section we describe our third algorithm for computing the
visibiliy map. It uses a complementary approach to the radial sweep in
the previous sections and sweeps the terrain centrifugally, by growing
a region $R$ around the viewpoint.  This region is kept
\emph{star-shaped} around~$v$: for any point $u$ inside $R$, the line
segment from $(u_x,u_y)$ to $(v_x,v_y)$ lies entirely inside~$R$. The
idea is to grow $R$ point by point until it covers the complete grid,
while maintaining the horizon $\sigma_R$ of $R$.  Recall that the
horizon $\sigma_R$ is a function from azimuth angles to elevation
angles, such that $\sigma_R(\alpha)$ is the maximum elevation angle of
any point on the intersection of $R$ with the ray that extends from
$v$ in direction $\alpha$.

Whenever a new point $u$ is added to $R$, we decide whether it is
visible. The star shape of $R$ guarantees that all points along the
line of sight from $v$ to $u$ have already been added, so we can in
fact decide whether $u$ is visible by determining whether $u$ is
visible above the horizon of $R$ just before adding $u$ (see
Figure~\ref{fig:centrifugal}). The key to a good performance is to have
a way of growing $R$ that results in an efficient disk access pattern,
and to have an efficient way of maintaining the horizon structure.
Below we explain how to do this, given an elevation grid with a fixed
number of bytes per grid point.

For ease of description, we remind the reader our notation: We assume
that the viewpoint $v$ is in the center of the grid at
coordinates $(0,0,0)$ and the grid has size $(2m + 1) \times (2m +
1)$, where $m = (\sqrt n - 1)/2$.  The elevations of the grid points
are given in a two-dimensional array $Z$ that is ordered row by row,
with rows numbered from $-m$ to $m$ from north to south
and columns numbered from $- m$ to $m$ from west to
east. By $p(i,j)$ we denote the grid point $q = (q_x,q_y,q_z)$ in row
$i$ and column $j$ with coordinates $q_x = j$, $q_y = -i$ and $q_z =
Z_{ij}$; by \emph{cell $(i,j)$} we denote the square $D(p(i,j))$.

To maintain the horizon efficiently, we represent it by a grid model
itself: more precisely, it is maintained in an array $S$ of $32m$
slots, where slot $i$ stores the highest elevation angle in $R$ that
occurs within the azimuth angle range from $i\cdot 2\pi /32m$ to
$(i+1)\cdot 2\pi / 32m$.

For growing the region $R$ the idea is to do so cache-obliviously
using a recursive algorithm.  Initially we call this algorithm with
the smallest square that contains the full grid and whose width is a
power of two.
When called on a square of size larger than one, it makes recursive
calls on each of the four quadrants of the square, in order of
increasing distance of the quadrants from~$v$. For a square tile with
upper left corner $(i,j)$ and width $s$, this distance
$\mathrm{TileDist}(i,j,s)$ is the distance from $v$ to the closest
point of the tile. 
 This is determined as follows. Let $v_i, v_j$ be
the row and column of the viewpoint.

\begin{itemize}
\item when $i \leq v_i < i + s$ and $j \leq v_j + s$, then the tile contains~$v$, and $\mathrm{TileDist}(i,j,s) = 0$;
\item otherwise, when $i \leq v_i < i + s$, the tile intersects the row that contains $v$, and $\mathrm{TileDist}(i,j,s) =  \min(|i - v_i|, |i + s - 1 - v_i|)$;
\item otherwise, when $j \leq v_j < j + s$, the tile intersects the column that contains $v$, and $\mathrm{TileDist}(i,j,s) = \min(|j - v_j|, |j + s - 1 - v_j|)$;
\item otherwise, $\mathrm{TileDist}(i,j,s) = \min(|i - v_i| + |j - v_j|, |i - v_i| + |j + s - 1 - v_j|, |i + s - 1 - v_i| + |j - v_j|, |i + s - 1 - v_i| + |j + s - 1 - v_j|)$.
\end{itemize}

When called on a square of size~1, that is, a square that contains
only a single grid point~$p(i,j)$, we proceed as follows. We retrieve
the elevation $Z_{ij}$ of~$p(i,j)$ from the input file and compute its
azimuth angle
$\mathrm{AzimAngle}(p(i,j))$
and its elevation angle.
$\mathrm{ElevAngle}(p(i,j))$.
Then we check if $p(i,j)$ is visible: this is the case if and only if
$p(i,j)$ appears higher above the horizon than the current horizon in
the direction of $p(i,j)$; that is, if and only if
$\mathrm{ElevAngle}(p(i,j)) > S[\lfloor
\mathrm{AzimAngle}(p(i,j))/2\pi \cdot 32m\rfloor]$. The visibility of
$p(i,j)$ is recorded in the output grid $V$. Next we update the
horizon to reflect the inclusion of $p(i,j)$ in $R$. To this end we
check all slots in the horizon array whose azimuth angle range
intersects the azimuth angle range of cell $(i,j)$; let ${\cal
  A}(p(i,j))$ denote this set of slots. For each slot of ${\cal
  A}(p(i,j))$ that currently stores an elevation angle lower than
$\mathrm{ElevAngle}(p(i,j))$, we raise the elevation angle to
$\mathrm{ElevAngle}(p(i,j))$. We thus have the following algorithm:

\addvspace{\baselineskip}
\noindent
\textbf{Algorithm} \textsc{CentrifugalSweep}:\\[.25\baselineskip]
create horizon array $S[0..32m-1]$\\
\textbf{for }$k \leftarrow 0$ \textbf{to} $32m-1$ \textbf{do} $S[k] \leftarrow -\infty$\\
$s \leftarrow$ smallest power of two $\geq 2m+1$\\
$\textsc{SweepRecursively}(\mathit{-m,-m,s})$

\addvspace{\baselineskip}
\noindent
\textbf{Algorithm} $\textsc{SweepRecursively}(i,j,s)$:\\
\textit{(Recursively computes visibility for the tile with upper left cell $(i,j)$ and width $s$)}
\\[.25\baselineskip]
\textbf{if }$s = 1$\\
\textbf{then }$\alpha \leftarrow \mathrm{AzimAngle}(p(i,j))$\\
\hphantom{\textbf{then }}$\beta \leftarrow \arctan(Z_{ij}/\mathrm{Dist}(i,j))$\hfill($=\mathrm{ElevAngle}(p(i,j))$)\\
\hphantom{\textbf{then }}\textbf{if} $\beta > S[\lfloor \alpha/2\pi \cdot 32m\rfloor]$ \textbf{then} $V_{ij} \leftarrow 1$ \textbf{else} $V_{ij} \leftarrow 0$\\
\hphantom{\textbf{then }}$\alpha^{-} \leftarrow$ smallest azimuth of any corner of cell $(i,j)$\\
\hphantom{\textbf{then }}$\alpha^{+} \leftarrow$ largest azimuth of any corner of cell $(i,j)$\\
\hphantom{\textbf{then }}\textbf{for }$k \leftarrow \lfloor \alpha^{-}/2\pi \cdot 32m\rfloor$ \textbf{to} $\lceil \alpha^{+}/2\pi \cdot 32m\rceil-1$\\
\hphantom{\textbf{then }}\textbf{do }$S[k] \leftarrow \max(S[k], \beta)$\\
\textbf{else }\textit{Let $Q$ be the four subquadrants:}\\
\hphantom{\textbf{then }}$s \leftarrow s/2$\\
\hphantom{\textbf{then }}$Q \leftarrow \{\langle i,j,s\rangle,\langle i+s,j,s\rangle,\langle i,j+s,s\rangle,\langle i+s,j+s,s\rangle\}$\\
\hphantom{\textbf{then }}sort the elements $\langle i,j,s\rangle$ of $Q$ by incr. $\mathrm{TileDist}(i,j,s)$\\
\hphantom{\textbf{then }}\textbf{for }$\langle i,j,s\rangle \in Q$\\
\hphantom{\textbf{then }}\textbf{do }$\textsc{SweepRecursively}(i,j,s)$

\begin{figure}[t]
  \centering{
    \includegraphics[height=2in]{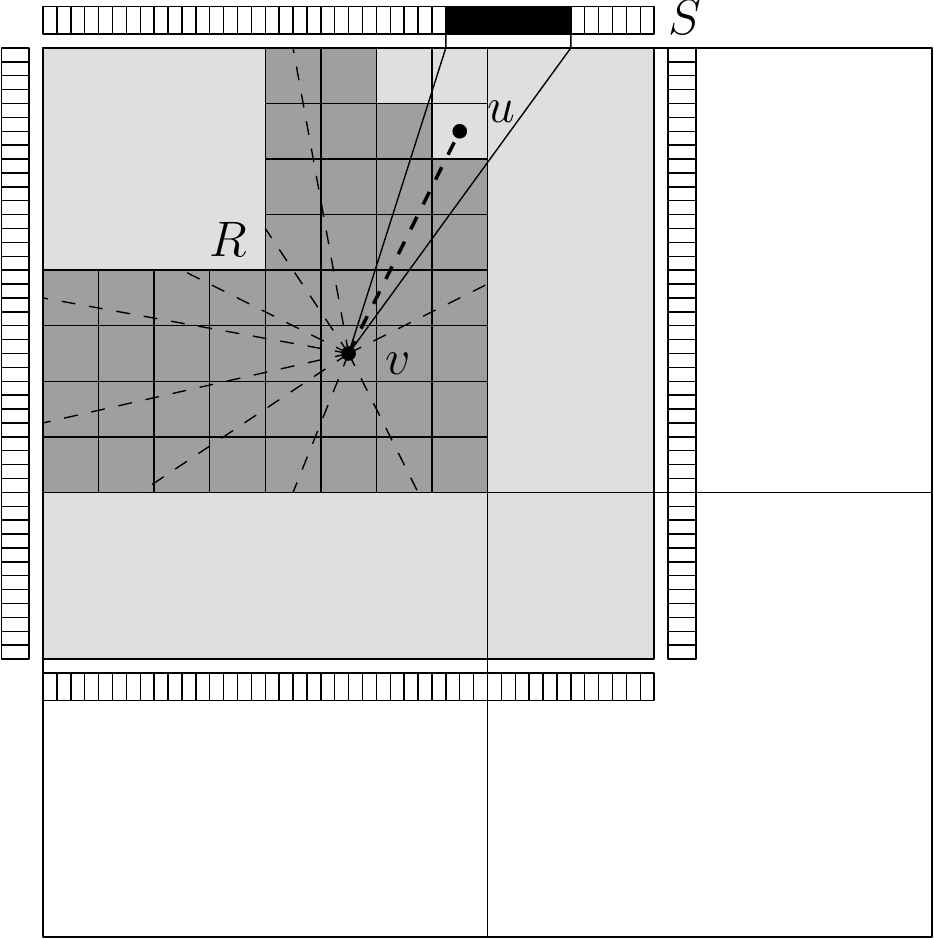}
    \caption{$R$ (dark shade) is a star-shaped region of the terrain (light
      shade) around $v$. 
      The horizon of $R$ is maintained in array $S$. When $u$ is added to
      $R$, the elevation angles in the black slots of $S$ are updated.
    }
    \label{fig:centrifugal}
  }
\end{figure}

\subsection{Accuracy of the centrifugal sweep}
Note that when the algorithm updates the horizon array, the elevation
angle of $p(i,j)$ may be used to raise the elevation angles of a set
of horizon array slots ${\cal A}(p(i,j))$, of which the total azimuth
range may be slightly larger than that of the cell corresponding to
$p(i,j)$---this is due to the rounding of the azimuth angles
$\alpha^{-}$ and $\alpha^{+}$ in the algorithm.  However, this is not
a problem:
The azimuth angles of grid points that lie next to each other (as seen
from the viewpoint) differ by at least roughly $1/m$. The size of the
horizon array is chosen such that its horizontal resolution is more
than four times bigger: it divides the full range of azimuth angles
from $0$ to $2\pi$ over $32m$ slots, each of which covers an azimuth
angle range of $2\pi / 32m < 1/4m$. Therefore, if the resolution of
the horizon array would be insufficient, then surely the resolution of
the original elevation grid would not be sufficient.
%

\subsection{Efficiency of the centrifugal sweep}
The number of recursive calls made by the region-growing algorithm is
$O(n)$. The only part of any recursive call that takes more than
constant time is the updating of the horizon. We analyse this layer by
layer, where this time layer $l$ is defined as the cells $(i,j)$ such
that $|i| + |j| = l$. There are $O(\sqrt n)$ layers, and on each
layer, each of the $O(\sqrt n)$ slots of the horizon array is
updated at most twice. Thus the total time for updating the horizon is
$O(n)$, and the complete algorithm runs in $O(n)$ time.

The number of I/Os under the tall-cache assumption ($M = \Omega(B^2)$)
can be analysed as follows. Let $w$ be the largest power of two such
that the elevation and visibility values of a square tile of 
$w \times w$ points fit in half of the main memory. There are $O(n/w^2)
= O(n/M)$ recursive calls on tiles of this size, and for each
of them the relevant blocks of the input and output files can be
loaded in $O(w(w/B + 1)) = O(w^2/B + w) = O(M/B + \sqrt M) = O(M/B)$
I/Os. Thus all I/O for reading and writing blocks of the input and
output files can be done in $O(n/M \cdot M/B) = O(\scan(n))$
I/Os in total.

It remains to discuss the I/Os that are caused by swapping parts of
the horizon array in and out of memory. To this end we distinguish (i)
recursive calls on tiles of size $w \times w$ at distance at least
$c \cdot \sqrt{n/M}$ from the viewpoint (for a suitable constant~$c$),
and (ii) calls on the remaining tiles around the viewpoint.  For
case (i), observe that each tile $G$ of size $w \times w$ at
distance at least $c \cdot \sqrt{n/M}$ from the viewpoint has an
azimuth range of $O(w/\sqrt{n/M}) = O(M/\sqrt{n})$; since the horizon
array has $O(\sqrt n)$ slots, $G$ spans $O(M/\sqrt{n} \cdot \sqrt n) =
O(M)$ slots of the horizon array. Therefore, when $c$ is sufficiently
large, the part of the horizon array that is relevant to the call on
$G$ can be read into the remaining half of the main memory at once,
using $O(\scan(M))$ I/Os. In total we get $O(n/M) \cdot O(\scan(M)) =
O(\scan(n))$ I/Os for reading and writing the horizon array in
instances of case~(i).  For case (ii), note that we access the horizon
array $O(n)$ times in total (as shown in our running time analysis
above). Because the tiles of case (ii) contain only $O(n/M)$ grid
points in total, the accesses to the horizon array are organised in
$O(n/M)$ runs of consecutive horizon array slots. The total number of
I/O-operations induced by these accesses is therefore $O(n/B + n/M) =
O(\scan(n))$.

Adding it all up, we find that the centrifugal sweep algorithms runs
in $O(n)$ time and $O(\scan(n))$ I/Os. The algorithm does not use or
control $M$ and $B$ in any way: it is cache-oblivious. The
I/O-efficiency analysis for the maintenance of the horizon array is
purely theoretical as far as disk I/O is concerned: the complete
horizon array easily fits in main memory for files up to several
trillion grid points. However, the I/O-efficiency analysis also
applies to the transfer of data between main memory and smaller
caches.

\section{An IO-efficient  algorithm using linear interpolation} 
\label{sec:algorithms}

In this section we describe our last two algorithms for computing
viewsheds, \visiter and \visdac. These algorithms use linear
interpolation to evaluate the intersection of the line-of-sight with
the grid lines, and constitute an improved, IO-efficient
version of Franklin's \rrr algorithm.

{\bf Notation.} Recall that the horizon $H_v$ (wrt to viewpoint $v$) is the upper rim
of the terrain as it appears to a viewer at $v$.  Suppose we recenter
our coordinate system such that $v = (0,0,0)$, and consider a
\emph{view screen} around the viewer that consists of the Cartesian
product of the vertical axis and the square with vertices $(1, 0)$,
$(0, 1)$, $(-1, 0)$ and $(0, -1)$. The projection of a point $p =
(p_x, p_y, p_z)$ towards $v$ onto the view screen has coordinates: $p
/ (|p_x| + |p_y|)$. Note that any line segment that does not cross the
north-south or east-west axis through $v$, will appear as a line
segment in the projection onto the view screen. We now define the
horizon of the terrain as it appears in the projection. More
precisely, for $t \in [0,2]$, we define the horizon $H_v(t)$ as the
maximum value of $p_z / (|p_x| + |p_y|)$ over all terrain points $p$
such that $p_x / (|p_x| + |p_y|) = 1 - t$ and $p_y \leq 0$ (this
defines the horizon of the terrain south of the viewpoint). For $t \in
[2,4]$, we define the horizon $H_v(t)$ as the maximum value of $p_z /
(|p_x| + |p_y|)$ over all terrain points $p$ such that $p_x / (|p_x| +
|p_y|) = 3 + t$ and $p_y \geq 0$ (this defines the horizon of the
terrain north of the viewpoint).

{\bf The model.}  We consider two models, shown in
Figure~\ref{fig:model}: In the \emph{gridlines} model the grid points
are connected by vertical and horizontal lines in a grid, and
visibility is determined by evaluating the intersections of the LOS
with the grid lines.  The gridlines model is the model used by
\rrr. We also consider a slightly different model, the \emph{layers}
model, in which the grid points are connected in concentric layers
around the viewpoint and visibility is determined by evaluating the
intersections of the LOS with the layers.  The layers model is a
relaxation of the gridlines model because it considers only a subset
of the intersections (obstacles) considerd by the gridlines model; any
point visible from $v$ in the grid model is also visible in the layers
model (but not the other way), and the viewshed generated by the grid
model is a subset of the viewshed generated with the layers model.

\begin{figure}[tb]
\label{fig:model}
\centering{
\begin{tabular}{c c}
 \includegraphics[width=3.5cm]{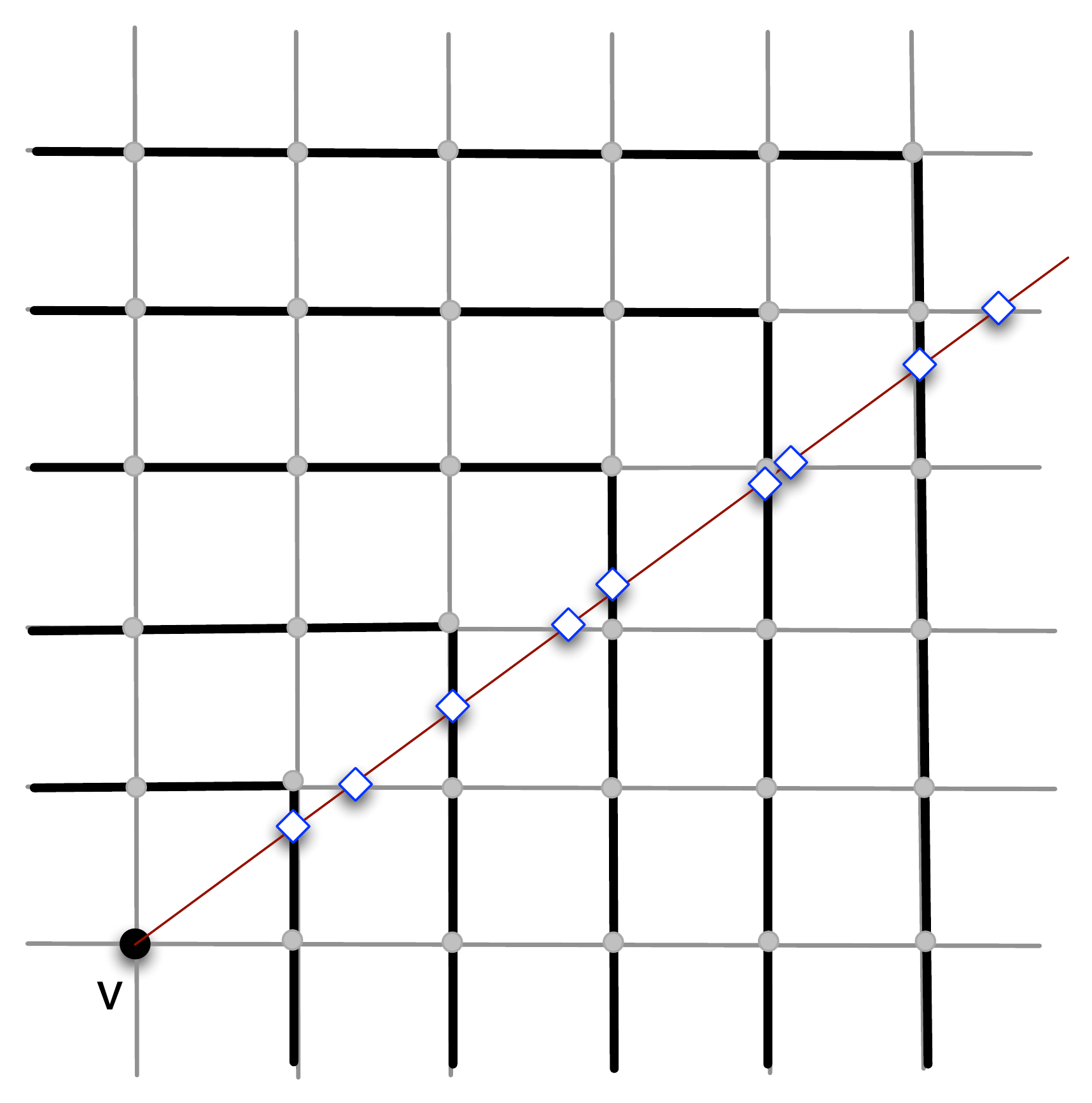} &
 \includegraphics[width=3.5cm]{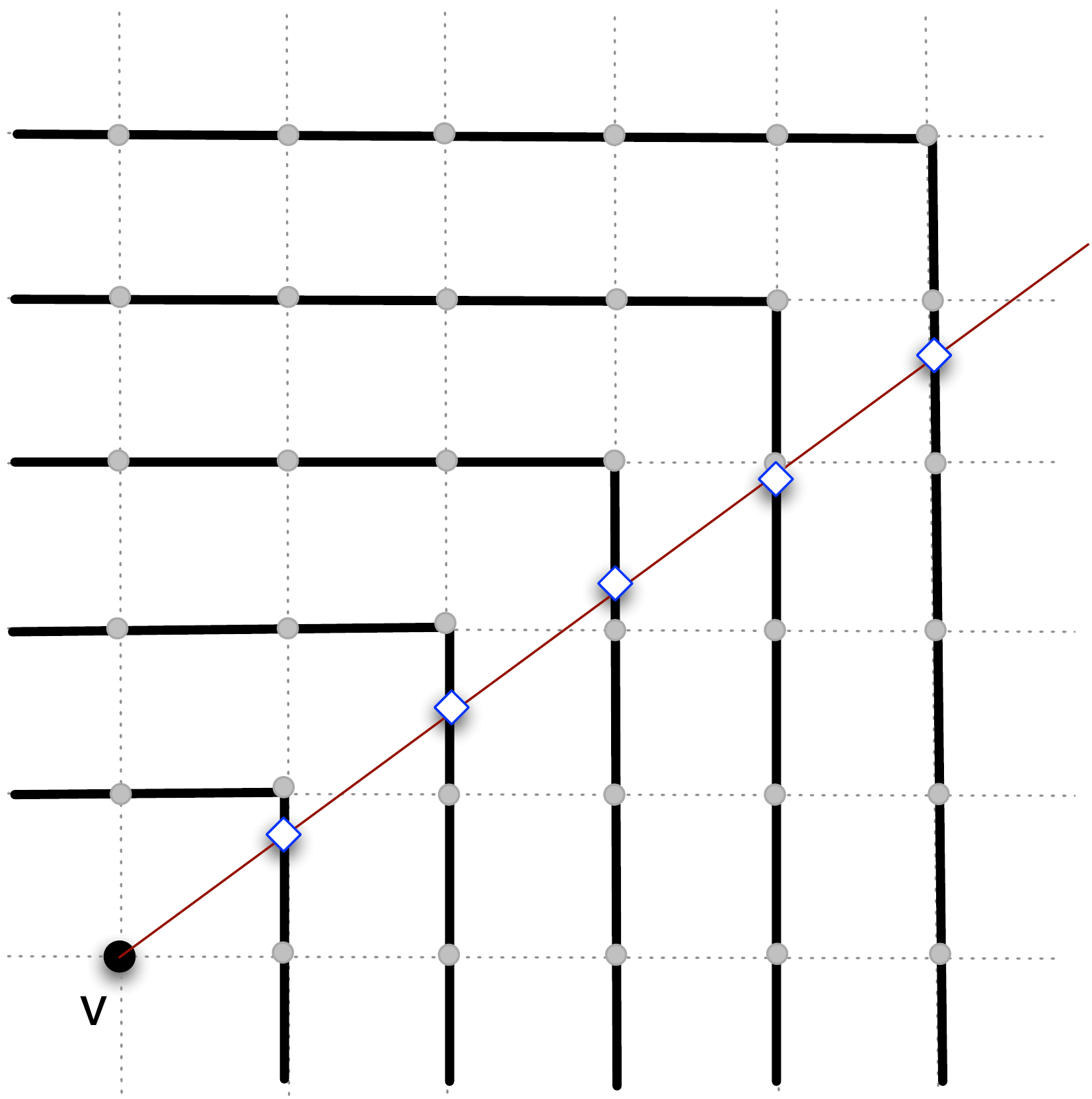} \\
(a) & (b)\\
\end{tabular}
\caption{\small{(a) The gridlines model: visibility
    is determined by the intersections of  the LOS with all the grid
    lines. (b) The layers model: visibility is determined by the intersections of the LOS with the layers. }}
}
\end{figure}

{\bf General idea and comparison to previous algorithms.} Our algorithms \visdac and
\visiter use an overall approach that is a combination of our radial
sweep algorithm (Section~\ref{sec:radial-layers}) which partitions the
grid into bands, and our centrifugal sweep algorithm
(Section~\ref{sec:conc-sweep}) which traverses the grid in layers
around the viewpoint and maintains the horizon of the region traversed
so far.
\begin{itemize} 
\item Recall that the radial sweep algorithm from
  Section~\ref{sec:radial-layers} consists of three phases: (1) partition
  the elevation grid in bands; (2) rotate the ray and compute
  visibility bands; (3) sort the visibility bands into a visibility grid.
  Phase 2 accesses data sequentially from all bands while the ray
  rotates around the viewpoint. The width of a band is chosen $w =
  \Theta(\sqrt M)$.
  Our algorithms \visiter and \visdac have the same first and third
  phase, but in phase (2) they process the bands one at a time. The
  size of a band is set so that a band fits fully in memory.

\item The centrifugal sweep algorithm Section~\ref{sec:conc-sweep}
  uses horizons which are stored discretized in an array.
  Our algorithms \visiter and \visdac use linear interpolation and
  therefore horizons are piecewise-linear functions and are stored
  as a list of $\{azimuth, zenith \}$ pairs with full precision. 

\end{itemize}

We start by describing how to compute viewsheds in the layers model in
Section~\ref{sec:visiter} and \ref{sec:visdac}; in
Section~\ref{sec:gridlines} we show how our algorithms can be
extended to the gridlines model while maintaining the same worst-case
complexity.

\subsection{An iterative algorithm: \visiter }
\label{sec:visiter}

This section describes our first viewshed algorithm in the layers
model, \visiter.  The main idea of \visiter is to traverse the grid in
layers around the viewpoint, one layer at a time, while maintaining
the horizon of the region traversed so far. The horizon is maintained
as a sequence of points $(t,H_v(t))$, sorted by $t$-coordinate,
between which we interpolate linearly. When traversing a point $p$,
the algorithm uses the maintained horizon to determine if $p$ is
visible or invisible.  In order to do this IO-efficiently, it divides
the grid in bands around the viewpoint and processes one band at a
time.  The output visibility grid is generated band by band, and is
sorted into a grid in the final phase of the algorithm.  The size of
the band is chosen such that a band fits in memory.  Below we explain
these steps in more detail.

{\bf Notation.}  The notation is the same as before and we review it
for clarity: we ssume that the viewpoint $v$ is in the center of the
grid at coordinates $(0,0,0)$ and the grid has size $(2m+1) \times
(2m+1)$, where $m = (\sqrt n -1) /2$. The elevations of the grid
points are given in a two-dimensional matrix $Z$ that is ordered row
by row, with rows numbered from $-m$ to $m$ from north to south, and
columns numbered from $-m$ to $m$ from west to east.  By $p(i,j)$ we
denote the grid point $q=(q_x, q_y, q_z)$ in row $i$ and column $j$
with coordinates $q_x = j, q_y = -i$ and $q_z= Z_{ij}$.

For $l\ge 0$, let \emph{layer} $l$ of the grid, denoted $L_l$, denote
the set of grid points whose $L_\infty$-distance from the viewpoint,
measured in the horizontal plane, is $l$.  By definition, $L_0$
consists of only one point, $v$.  We divide the grid in concentric
bands around the viewpoint: For $k>0$, band $k$,
denoted $B_k$, contains all points in layers $w_{k-1}$ (inclusive)
to $w_k$ (exclusive), where $w_k (k \ge 0)$ denote the indices of
layers corresponding to the band boundaries. Thus band $B_1$ contains
all points in layers $w_0=1$ to $w_1$,
and so on.

The algorithm starts with a preprocessing step which, given an
arbitrary constant $K$, computes the band boundaries $w_k (k \ge 1)$
such that a band has size $\Theta(K)$ as follows: it cycles through
each layer $i$ in the grid, computes (analytically) the number of
points in that layer, and checks whether including this layer in the
current band makes the band go over $K$ points.  If yes, then layer
$i$ marks the start of the next band. Otherwise, it adds the points in
layer $i$ to the current band and continues.

The maximum size $K$ of a band is chosen such that a band fills
roughly a constant fraction of memory, and each band is at least one
layer wide.  More precisely, we choose $K = c_1 M$ and assume $\sqrt n
\leq c_1 M$, for a sufficiently small constant $c_1$ which will be
defined more precisely later.  Thus the number of bands, $N_{bands}$,
is $O(N/M)$.

Once the band boundaries are set, the algorithm proceeds in three
phases. The first phase is to generate, for each band $B_k$, a list
$E_k$ containing the elevations of all points in the band.  It does
this by scanning the grid in row-column order: for each point
$p(i,j)$, it calculates the index $k$ of the band that contains the
point and writes $Z_{ij}$ to $E_k$.  We note that the first phase
writes the lists $E_k$ sequentially, and thus list $E_k$ contains the
points in the order in which they are encountered during the
(row-by-row) scan of the grid.  The algorithm is given below.

\noindent
\textbf{Algorithm} \textsc{BuildBands}:\\[.25\baselineskip]
load list containing band boundaries in memory \\
\textbf{for }$k \leftarrow 1$ \textbf{to} $N_{bands}$\\
\textbf{do }initialize empty list $E_k$\\
\textbf{for } $i\leftarrow -m$ \textbf {to} $m$\\
\textbf{do }  \textbf{for } $j\leftarrow -m$ \textbf {to} $m$\\
\hphantom{\textbf{do }}  \textbf{do} read next elevation $Z_{ij}$ from grid \\
\hphantom{\textbf{do }\textbf{do }} $k \leftarrow  \textrm{band containing
point } (i,j)$\\
\hphantom{\textbf{do }\textbf{do } }append $Z_{ij}$ to $E_k$


Given the lists $E_k$, the second phase of the algorithm computes
which points are visible.  To do this it traverses the grid one band
at a time, reading the list $E_k$ into memory. Once a band is in
memory, it traverses it layer by layer from the viewpoint outward,
counter-clockwise in each layer.  The output of the second phase is a
set of lists $V_k$ with visibility values, one list for each band.
While traversing the grid in this fashion the algorithm maintains the
horizon of the region encompassed so far. More precisely, let $L_{1,i}
(i\ge 1)$ denote the set of points in layers $L_1$ through
$L_i$. Before traversing the next layer $L_{i+1}$, the algorithm knows
the horizon $H_{1,i}$ of $L_{1,i}$.  While traversing the points in
$L_{i+1}$, the algorithm determines for each point $p$ if it is above
or below the horizon $H_{1,i}$ and records this in $V_k$. At the same
time it updates $H_{1,i}$ on the fly to obtain $H_{1,i+1}$. To do so,
the algorithm computes, for each point $p$, the projection $h$ onto
the view screen of the line segment that connects $p$ to the previous
point in the same layer, the algorithm computes the intersection of
$h$ with the current horizon as represented by $H_{1,i}$, and then
updates $H_{1,i}$ to represent the upper envelope of the current
horizon and $h$.  After traversing the entire grid in this manner, the
set of points that have been marked visible during the traversal
constitute the viewshed of $v$. The algorithm is given below only for
the first octant; the other octants are handled similarly:

\noindent
\textbf{Algorithm} \textsc{VisBands-ITER}:\\[.25\baselineskip]
$H_{1,0} \leftarrow \emptyset$\\
\textbf{for }$k \leftarrow 1$ \textbf{to} $N_{bands}$ \\
\textbf{do }                    load list $E_k$ in memory\\
\hphantom{\textbf{do }}                   create list $V_k$ in memory
and initialize it as
all invisible\\
\hphantom{\textbf{do }} \textbf{for }  $l \leftarrow  w_{k-1}$ to
$w_k$   //for each layer in the band \\
\hphantom{\textbf{do }} \textbf{do }  //traverse layer $l$ in ccw order\\
\hphantom{\textbf{do } \textbf{do }} \textbf{for } $r \leftarrow 0$
to $-l$ //first octant\\ 
\hphantom{\textbf{do } \textbf{do }} \textbf{do } get elevation
$Z_{rl}$ of $p(r,l)$ from $E_k$\\
\hphantom{\textbf{do } \textbf{do } \textbf{do }} determine if
$Z_{rl}$  is
above   $H_{1,l-1}$\\
\hphantom{\textbf{do } \textbf{do } \textbf{do }} if visible, set
value $V_{rl}$  in $V_k$ as visible\\
\hphantom{\textbf{do } \textbf{do } \textbf{do }}  $h \leftarrow$ projection of $p(r-1,l)p(r,l)$\\
\hphantom{\textbf{do } \textbf{do } \textbf{do }}  merge $h$ into horizon $H_{1,l-1}$\\
 \hphantom{\textbf{do } \textbf{do }} $H_{1,l} \leftarrow H_{1,l-1}$\\


The third and final phase of the algorithm creates the visibility grid
$V$ from the lists $V_k$.  We note that in phase~2 the lists $V_k$ are
stored in the same order as $E_k$, that is, the order in which the
points in the band are encountered during a row-by-row scan of the
grid; keeping points in this order is convenient because it saves an
additional sort, and in the same time this is precisely the order in
which they are needed by phase 3.  Phase 3 is the reverse of phase 1:
for each point $(i,j)$ in the grid in row-major order, it computes the
band $k$ where it falls, accesses list $V_k$ to retrieve the
visibility value of point $(i,j)$, and writes this value to the output
grid $V$.
The crux in this phase is that it simply reads the lists $V_k$
sequentially. The algorithm is given below:

\noindent
\textbf{Algorithm} \textsc{CollectBands}:\\[.25\baselineskip]
load list containing band boundaries in memory\\
initialize empty list $V$\\
\textbf{for } $i\leftarrow -m$ \textbf {to} $m$\\
\textbf{do }  \textbf{for } $j\leftarrow -m$ \textbf {to} $m$\\
\hphantom{\textbf{do }}  \textbf{do} $k \leftarrow \textrm{band
  containing point } (i,j) $\\
\hphantom{\textbf{do }\textbf{do }} get value $V_{ij}$ of
point $(i,j)$  from $V_k$\\
\hphantom{\textbf{do }\textbf{do } }append  $V_{ij}$ to list $V$


{\bf Efficiency analysis of \visiter.}
We now analyze each phase in \visiter under the assumption that $n \leq c M^2$ for a 
sufficiently small constant $c$. 
The pre-processing phase runs in $O(n)$ time and no \io (does not
access the grid).  The output of this step is a list of $O(N_{bands})
= O(n/M)$ band boundaries, which fits in memory assuming that $n \leq
c M^2$ for a sufficiently small constant $c$.

The first phase, \textsc{BuildBands}, reads the points of the
elevation grid in row-column order, which takes $O(n)$ time and
$O(\scan(n))$ \ios.  With the list of band boundaries in memory, the
band containing a point $(i,j)$ can be computed with, for example, binary
search in $O(\lg n/M) = O(\lg n)$ time and no \io.  The lists $E_k$
are written to in sequential order.
If one block from each band fits in memory, which happens when 
$n \leq c M^2 / B$ for a sufficiently small constant $c$ (so that $N_{bands} = O(n/M) = O(M/B)$), 
then writing the lists $E_k$ directly takes $O(\scan(n)) = O(\sort(n))$ \ios (note that $O(\sort(n))$ and $O(\scan(n))$ are equal if $n = O(M^2/B)$).
If we cannot keep one block of each band in memory, that is, $n > c M^2/B$, 
then we perform a hierarchical distribution as follows: we group the
$N_{bands}$ bands in $O(M/B)$ super-bands, keep a write buffer of one block for each super-band in
memory, distribute the points in the grid to these super-bands, and
recurse on the super-bands to distribute the grid points to individual bands. 
A pass takes $O(\scan(n))$ \ios, overall it takes
$O(\log_{M/B} N_{bands}) = O(\log_{M/B} N/M)$ passes, and thus the
first phase has I/O-complexity $O(\sort(n))$.
In total, the first phase takes $O(n \lg n)$ time and $O(\sort(n))$ I/Os.

The second phase, \textsc{VisBands-ITER}, takes as input the lists $E_k$
and computes the visibility bands $V_k$. We choose $K = c_1 M$ such that the elevations $E_k$ and the visibility map $V_k$ of any band $B_k$ of size $K$ fits in 2/3 of the memory; the remaining 1/3 of the memory is saved for the horizon structure. While processing a band $B_k$ in the second phase, the points in $E_k$ and $V_k$ are not accessed sequentially. However, given the band boundaries, the location of any point in a band can be
determined analytically, and thus the value (elevation or visibility)
of any point in a band can be accessed in constant time, without any
search structure, and without any \io.
Let us denote by $H_{tot}$ the total cumulative size of all partial
horizons $H_{1,l}$: $H_{tot} = \sum_{l=1}^{\sqrt n} |H_{1,l}|$.  The
horizon $H_{1,l}$ is maintained as a list $\{ (t, h)\}$ of horizontal
and vertical coordinates on the view screen, sorted counter-clockwise
(ccw) around the viewpoint.  As the algorithm traverses a layer $l$ in
ccw order, it also traverses $H_{1,l-1}$ in ccw order, and constructs
$H_{1,l}$ in ccw order. To determine whether a point is above the
horizon, it is compared with the last segment in the horizon; if the
point is above the horizon, it is added to the horizon. Thus the
traversal of a layer $l$ runs in $O(|L_l| + |H_{1,l-1}| + |H_{1,l}|)$
time.  Over the entire grid, phase~2 runs in $O(\sum_l(|L_l| +
|H_{1,l}|)) = O(n + H_{tot})$~time.
The IO-complexity of the second phase: The algorithm reads
$E_k$ into memory, and writes $V_k$ to disk at the end. Over all the
bands this takes $O(\scan(n))$ \ios. If the horizon $H_{1,l}$ is small
enough so that it fits in memory (for any $l$), then accessing the
horizon does not use any IO. If the horizon does not fit in memory, we
need to add the cost of traversing the horizon in ccw order, for every
layer,   $O(\scan(\sum_{l=1}^{\sqrt n} |H_{1,l}|)) = \scan(H_{tot})$ \ios.

Finally, the third phase, \textsc{CollectBands}, takes as input the
lists $V_k$ and the list of band boundaries and writes the visibility
map. For $n \leq c M^2$, the list of band boundaries fits in memory. For
any point $(i,j)$ the band containing it can be computed in $O(\lg n)$
time and no IO.  The bands $V_k$ store the visibility values in the
order in which they are encountered in a (row-column) traversal of the
grid. Thus, once the index $k$ of the band that contains point $(i,j)$
is computed, the visibility value of this point is simply the next
value in $V_k$. As with step 1, we distinguish two cases: if the
number of bands is such that one block from each band fit in memory,
then this step runs in $O(n)$ time and $O(\sort(n)) = O(\scan(n))$ \ios.  
Otherwise, this step first performs a multi-level $M/B$-way merge of 
the bands into $O(M/B)$ super-bands so that one block from each can 
reside in main memory; in this case, the complexity of the step is 
$O(n \lg n)$ time and $O(\sort(n))$ \ios.
%
%
Putting everything together, we have the following:
\begin{theorem} \label{th:visiter}
  The algorithm \visiter computes viewsheds in the layers model in
  $O(n \lg n + H_{tot})$ time and $O(\sort(n) + \scan(H_{tot}))$ \ios,
  provided that $n \leq c M^2$ for a sufficiently small constant $c$.
\end{theorem}

Furthermore, if $n= O(M^2/B)$ and the partial horizons $H_{1,l}$ are
small enough to fit in memory for any $l$, the overall IO complexity
becomes $O(\scan(n))$~\ios.
We note that when $n = \Omega(M^2)$ the algorithm can be adapted using
standard techniques to run in the same bounds from
Theorem~\ref{th:visiter}; we do not detail on this because it has no
relevance in practice.

{\bf Discussion.}
Phase 1 and 3 of the algorithm are very simple and perform a scanning
pass over the grid and the bands, provided that $n \leq c M^2/B$: Phase 1
reads the input elevation grid sequentially and writes the elevation
bands sequentially; Phase~3 reads the visibility bands sequentially
and writes the visibility grid sequentially.  We found this condition
to be true in practice on our largest test grid ($28GB$) and with as
little as $.5GB$ of RAM. With more realistic value of $M=8 GB$ (and
$B=16KB$), the condition is true for $n$ up to $10^{15}$ points.
Thus, handling the sub-case $n \leq c M^2/B$ separately in the algorithm
provides a simplification and a speed-up without restricting
generality.

Phase 2, which scans partial horizons $H_{1,l}$ for every layer, runs
in $O(n + H_{tot})$ time and $O(\scan(n + H_{tot}))$ \ios.  
As we will prove below, in the worst case $|H_{1,l}| = \Theta(l^2)$, and the running time of the
second phase could be as high as $O(H_{tot}) = O(\sum_{l=1}^{O(\sqrt n)} \Theta(l^2)) = O(n \sqrt n)$,
with handling the horizon dominating the running time.  The worst-case
complexity is high but, on the other hand, if $H_{1,l}$ are small,
they fit in memory and the algorithm is fast. In particular if
$H_{tot}$ is $O(n)$, then phase 2 is linear.  This seems to be the
case on all terrains and all viewpoints that we tried and may be a
feature of realistic terrains. In
Section~\ref{sec:results} we'll discuss our empirical findings in
more detail.

\newcommand{\xwd}[1]{\mathrm{width}(#1)} {\bf Worst-case complexity of
  the horizon.}  Since the horizon is the upper envelope of the
projections of grid line segments onto the view screen, its complexity
is at most $O(n \alpha(n))$, where $n$ is the number of line
segments~\cite{HS86,WS88}. We will now show that we can prove a better
bound for our setting. Let the \emph{width} $\xwd{s}$ of a line
segment $s$ be the length of its projection on a horizontal line. We
need the following lemma.
\begin{figure}
\centering
\includegraphics[width=0.8\hsize]{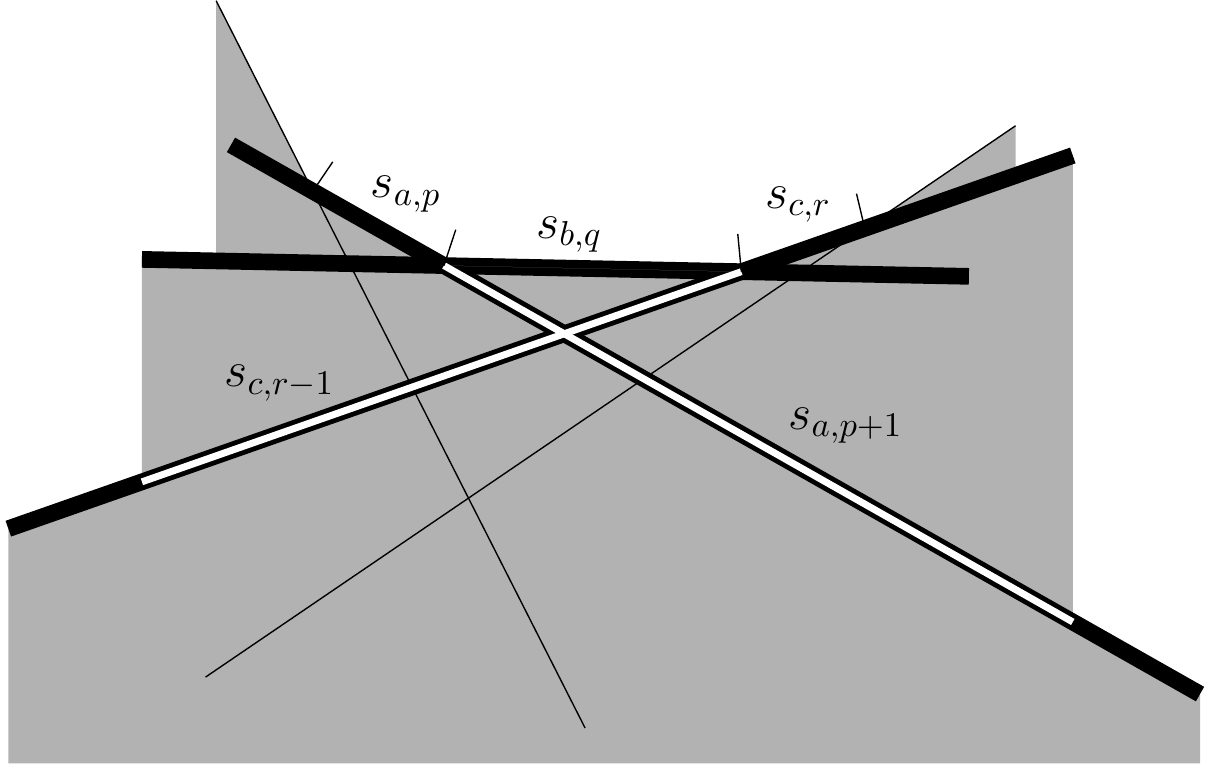}
\caption{Illustration of the proof of Lemma~\ref{lem:boundedwidth}. The white subsegments take the charge for $s_{b,q}$; together they are at least as wide as $s_b$, since they stick out from under $s_b$ on both sides.}
\label{fig:boundedwidth}
\end{figure}

\begin{lemma}\label{lem:boundedwidth}
If $S$ is a set of $n$ line segments in the plane, such that the widths of the line segments of $S$ do not differ in length by more than a factor $d$, then the upper envelope of $S$ has complexity $O(dn)$.
\end{lemma}

\begin{proof}
Let $s_1,...,s_n$ be the segments of $S$. Each segment $s_i$ consists of a number of maximal subsegments such that the interior of each subsegment lies either entirely on or entirely below the upper envelope. Let the subsegments of $s_i$ be indexed by $s_{i,j}$, such that the subsegments of $s_i$ from left to right are indexed by consecutive values of $j$, and such that $s_{i,j}$ is part of the upper envelope if and only if $j$ is odd. Let $u_1,...,u_m$ be the line segments of the upper envelope.

We consider two categories of line segments on the upper envelope: (i) segments that have at least one endpoint that is an endpoint of a segment of $S$; (ii) segments whose endpoints are no endpoints of segments in $S$.

Clearly, there can be only $O(n)$ segments of category (i), one segment to the left of each endpoint of a segment in $S$ and one segment to the right of each endpoint.

We analyze the number of segments of category (ii) with the following
charging scheme. Given a segment $u_h = s_{b,q}$ of category (ii), let
$s_{a,p}$ be the segment $u_{h-1}$ and let $s_{c,r}$ be the segment
$u_{h+1}$. We charge $u_h$ to $s_{a,p+1}$ and $s_{c,r-1}$. Observe
that with this scheme, each segment $s_{i,j}$ can only be charged
twice, namely by the successor of $s_{i,j-1}$ on the upper envelope
and by the predecessor of $s_{i,j+1}$ on the upper envelope. Since
each segment $s_i$ has only one leftmost and only one rightmost
subsegment, and each is charged at most twice (in fact, once), there
are at most $O(n)$ segments of category (ii) that put charges on
leftmost or rightmost subsegments. If neither $s_{a,p+1}$ is the
rightmost subsegment of $s_a$ nor $s_{c,r-1}$ is the leftmost
subsegment of $s_c$, then $s_a$ must appear on the upper envelope
again somewhere to the right of the right end of $s_b$, and $s_c$ must
appear on the upper envelope again somewhere to the left of the left
end of $s_b$ (see Figure~\ref{fig:boundedwidth}) Therefore
$\xwd{s_{a,p+1}} + \xwd{s_{c,r-1}} > \xwd{s_b} \geq
\frac{1}{dn}\sum_{i=1}^n \xwd{s_i}$. Since each subsegment is charged
at most twice, the total length of subsegments charged is at most
$2\sum_{i=1}^n \xwd{s_i}$.  Thus there are less than $2dn$ segments of category (ii) that put charges on subsegments that are not leftmost or rightmost.
\end{proof}

Note that the widths of the projections of the edges of layer $l$ on
the view screen vary between $1/(l+1)$ and $1/(4l-2)$. Therefore, the
widths of the projections of the edges of the $l/2$ outermost layers
in a square region of $l$ layers around $v$ differ by less than a
factor 8. Thus, from Lemma~\ref{lem:boundedwidth} we get:
%

\begin{corollary}\label{cor:outerhorizon}
If $S$ consists of the $O(l^2)$ edges of the $l/2$ outermost layers in a square region of $l$ layers around $v$, then the horizon of $S$ has complexity $O(l^2)$.
\end{corollary}
%
%

\begin{lemma}\label{lem:combine}
If $S$ and $T$ are two $x$-monotone polylines of $m$ and $n$ vertices, respectively, then the upper envelope of $S$ and $T$ has at most $2(m+n)$ vertices.
\end{lemma}

\begin{proof}
There are two types of vertices on the upper envelope: vertices of $S$ or $T$, and intersection points between edges of $S$ and $T$. Clearly, there are at most $m + n$ vertices of the first type. Between any pair of vertices of the second type, there must be a vertex of the first type. Thus there are at most $m + n - 1$ vertices of the second type.
\end{proof}

\begin{theorem}\label{thm:gridhorizon}
If $S$ consists of $l$ layers in a square region around $v$, then the horizon of $S$ has complexity $O(l^2)$ in the worst case.
\end{theorem}
\begin{proof}
Let $T(l)$ be the complexity of the horizon of the innermost $l$ layers around $S$. 
By Lemma~\ref{lem:combine}, $T(l)$ is at most twice the complexity of the horizon of the innermost $l/2$ layers, plus twice the complexity of the remaining $l/2$ layers. By Corollary~\ref{cor:outerhorizon}, the latter is $O(l^2)$, and therefore we have $T(l) \leq 2T(l/2) + O(l^2)$. This solves to $T(l) = O(l^2)$.
\end{proof}

\subsection{A refined algorithm: \visdac}
\label{sec:visdac}

This section describes our second algorithm for computing viewsheds in
the layers model, \visdac. \visdac is a divide-and-conquer refinement
of \visiter and uses the same general steps: it splits the grid into
bands, computes visibility one band at a time, and creates the
visibility grid from the bands. The first phase (\textsc{BuildBands})
and last phase (\textsc{CollectBands}) are the same as in \visiter;
the only phase that is different is computing visibility in a band,
\textsc{VisBands-DAC}, which aims to improve the time to merge
horizons in a band using divide-and-conquer.
%
%

Similar to \textsc{VisBands-Iter}, \textsc{VisBands-DAC} processes the
bands one at a time: for each band $k$ it loads list $E_k$ in memory,
creates a visibility list $V_k$ and initializes it as all visible. It
then marks as invisible all points that are below $H_{prev}$, where
$H_{prev}$ represents the horizon of the first $k-1$ bands (more on
this below).  The bulk of the work in \textsc{visBands-DAC} is done by
the recursive function \textsc{Dac-Band}, which computes and returns the
horizon $H$ of $E_k$, and updates $V_k$ with all the points that are
invisible inside $E_k$.  This is described in detail below. Finally,
the horizon $H$ is merged with $H_{prev}$ setting it up for the next
band.

In order to mark as invisible the points in band $k$ that are below
$H_{prev}$ we 
first sort the points in the band by azimuth angle and then
scan them in this order while also scanning $H_{prev}$ (recall that
$H_{prev}$ is stored in ccw order).  Let $(a_1=0, h_1), (a_2,
h_2)$ be the first two points in the horizon $H_{prev}$.  For every
point $p=(a,h)$ in $E_k$ with azimuth angle $a \in [a_1, a_2]$, we
check whether its height $h$ is above or below the height of segment
$(a_1,h_1)(a_2,h_2)$ in $H_{prev}$.
When we encounter a point in $E_k$ with $a > a_2$, we proceed to the
next point in $H_{prev}$ and repeat.

The recursive algorithm \textsc{DAC-Band} takes as arguments an elevation
band $E_k$, a visibility band $V_k$, and the indices $i$ and $j$ of
two layers in this band ($w_{k-1} \le i \le j < w_k$). It computes
visibility for the points in layers $i$ through $j$ (inclusive) in
this band, and marks in $V_k$ the points that are determined to be
invisible during this process. In this process it also computes and
returns the horizon of layers $i$ through $j$ in this band.
\textsc{Dac-Band} uses divide-and-conquer in a straightforward way:
first it computes a ``middle'' layer $m, i \le m \le j$ between $i$ and
$j$ that splits the points in layers $i$ through $j$ approximately in
half. Then it computes visibility and the horizon recursively on each
side of $m$; marks as invisible all points in the second half that
fall below the horizon of the first half; and finally, merges the two
horizons on the two sides and returns the result.

\noindent
\textbf{Algorithm} \textsc{Dac-Band}($E_k,  V_k, i,j$):\\
[.25\baselineskip]
\textbf{if } $i==j$\\
\hphantom{\textbf{if }} $h \leftarrow$ compute-layer-horizon(i)\\
\hphantom{\textbf{if }} return $h$\\
\textbf{else } \\
\hphantom{\textbf{else }} $m \leftarrow$ middle layer between $i$ and $j$\\
\hphantom{\textbf{else }} $h_1 \leftarrow \textsc{Dac-Band}(E_k, V_k, i,m)$\\
\hphantom{\textbf{else }} $h_2 \leftarrow \textsc{Dac-Band}(E_k, V_k, m+1,j)$\\
\hphantom{\textbf{else }} mark invisible all points in $L_{m+1, j}$
that fall below $h_1$\\
\hphantom{\textbf{else }} $h \leftarrow$ merge($h_1,h_2$)\\
\hphantom{\textbf{else }} return $h$

{\bf Efficiency analysis of \visdac.}
The analysis of the first and last phase of \visdac,
\textsc{buildBands} and \textsc{CollectBands}, is the same as in
Section~\ref{sec:visiter}.  We now analyze
\textsc{VisBands-DAC}. Recall that we can assume that $E_k$ and $V_k$
both fit in memory during this phase (see
Section~\ref{sec:visiter}). The elevation and visibility of any point
in a band can be accessed in $O(1)$ time, without any search structure
and without any I/O.   We denote $H^B_{1,i}$ the horizon of (the
points in) the first $i$ bands; and by $H^B_{tot} = H^B_{1,1} +
H^B_{1,2} + H^B_{1,3} + ...= \sum_{i=1}^{N_{bands}} H^B_{1,i}$.

\begin{itemize} 
\item Marking as invisible the points in $E_k$ that are below
  $H_{prev}$ (here $H_{prev}$ represents $H^B_{1,k-1}$): this can be
  done by first sorting $E_k$ and then scanning $H^B_{1,k-1}$ and
  $E_k$ in sync. Over the entire grid, this takes $O(n \lg n +
  H^B_{tot})$ CPU and $O(\scan(n) + \scan(H^B_{tot}))$ \ios.

\item Merging horizons: After \textsc{Dac-band} is called in a band,
  the returned horizon is merged with $H_{prev}$. Two horizons can be
  merged in linear time and \ios. Over the entire grid this is
  $O(H^B_{tot})$ time and $O(\scan(H^B_{tot}))$~\ios.

\item \textsc{Dac-band}: This is a recursive function, with the
  running time given by the recurrence $T(k) = 2T(k/2) + \texttt{merge
    cost + update cost}$, where $k$ is the number of points in the
  slice between layers $i$ and $j$ given as input.  The base case
  computes the horizon of a layer $l$, which takes linear time wrt to
  the number of points in the layer.  Summed over all the layers in
  the slice the base case takes $O(\sum_{l=i}^{j}|L_l|) = O(k)$ time
  and no \io (band is in memory).

\item The update time in \textsc{Dac-band} represents the time to mark
  as invisible all points in the second half that fall below the
  horizon $h_1$ of the first half.  Recall that a band fits in memory
  and thus an input slice in \texttt{Dac-Band} fits in memory.  If the
  band is sorted, the update can be done as above in $O(k + |h_1|) =
  O(k)$ time (by Theorem~\ref{thm:gridhorizon} we have $|h_1| = O(k)$).
%

\item The merge time in \textsc{Dac-band} represents the time to merge
  the horizons $h_1$ and $h_2$ of the first and second half of the
  slice, respectively.  This takes $O(|h_1| + |h_2|) = O(k)$ time.

\item Putting it all together in the recurrence relation we get $T(k)
  = 2T(k/2) + O(k)$, which solves to $O(k \lg k)$ time.  Summed over
  all bands in the grid \textsc{Dac-Band} runs in $O(n \lg n)$ time
  and $O(\scan(n))$ \ios.

\end{itemize}

 Overall we have the following:

 \begin{theorem}
   The algorithm \visdac computes viewsheds in the layers model in
   $O(n \lg n + H^B_{tot})$ time and $O(\sort(n) + \scan(H^B_{tot}))$
   \ios, provided that $n \le cM^2$, for a sufficienly small constant
   $c$.
 \end{theorem}

 {\bf Discussion:} The worst case complexity of $H^B_{tot}$ is
 $\sum_{i=1}^{N_{bands}} H^B_{1,i} = O(N_{bands} \cdot n) =
 O(n^2/M)$; This is an improvement over $O(n \sqrt n)$ ( provided that
 $n \le cM^2$).
 Consider a band that extends from layer $L_i$ to layer $L_j$ and
 contains $k$ points. The algorithm \textsc{Dac-Band} runs in $O(k \lg
 k)$ time, while the iterative algorithm \textsc{VisBands-Iter} scans
 iteratively through all cumulative horizons of the layers in the band
 $H_{1,i}, H_{1,i+1}, ...$ and so on and runs in $O(k + |H_{1, i}| +
 |H_{1, i+1}| + ... +|H_{1,j}|)$.  When the horizons are small,
 \visiter runs in $O(k)$ time and is faster than \visdac. The
 divide-and-conquer merging is not justified unless the horizons are
 large enough to benefit from it.

\subsection{The gridlines model}
\label{sec:gridlines}

The algorithms \visdac and \visiter described in
Section~\ref{sec:visiter} and \ref{sec:visdac} above compute viewsheds
in the layers model.  
Let $X_i$ denote the line segments connecting points at distance $i-1$
with points at distance $i$ (Figure~\ref{fig:gridlines}). The set
$X_i$ represents the additional ``obstacles'' in the $i$th layer that
could intersect the LOS in the gridlines model.  With this notation
the horizon of the $i$th layer in the gridlines model is $H(L_i) \cup
H(X_i)$.
The algorithms \visiter and \visdac can be extended to compute
viewsheds in the gridlines model---the only difference is that they
compute the horizon of a layer as $H(L_i) \cup H(X_i)$ instead of
$H(L_i)$.  Since $|X_i| = \Theta(|L_i|)$, the analysis and the bounds
of the algorithms are the same in both models up to a constant factor.
Our algorithms \visiter and \visdac, when using the
gridlines model, compute the same viewshed as
\rrr~\cite{franklin:sdh94}.  \visdac's upper bound of $O(n \lg n +
n^2/M)$ is an improvement over \rrr's bound of $O(n \sqrt n)$,
provided that $n \le c M^2$.

The results on the worst-case complexity of the horizon in the layer
model extend to the gridlines model.  The extension is not entirely
straightforward, because the differences in width in the projection
between non-layer edges are larger than between layer edges.  We defer the proof to the journal version of this paper.


\begin{figure}[t]
  \centering{
  \includegraphics[width=3.5cm]{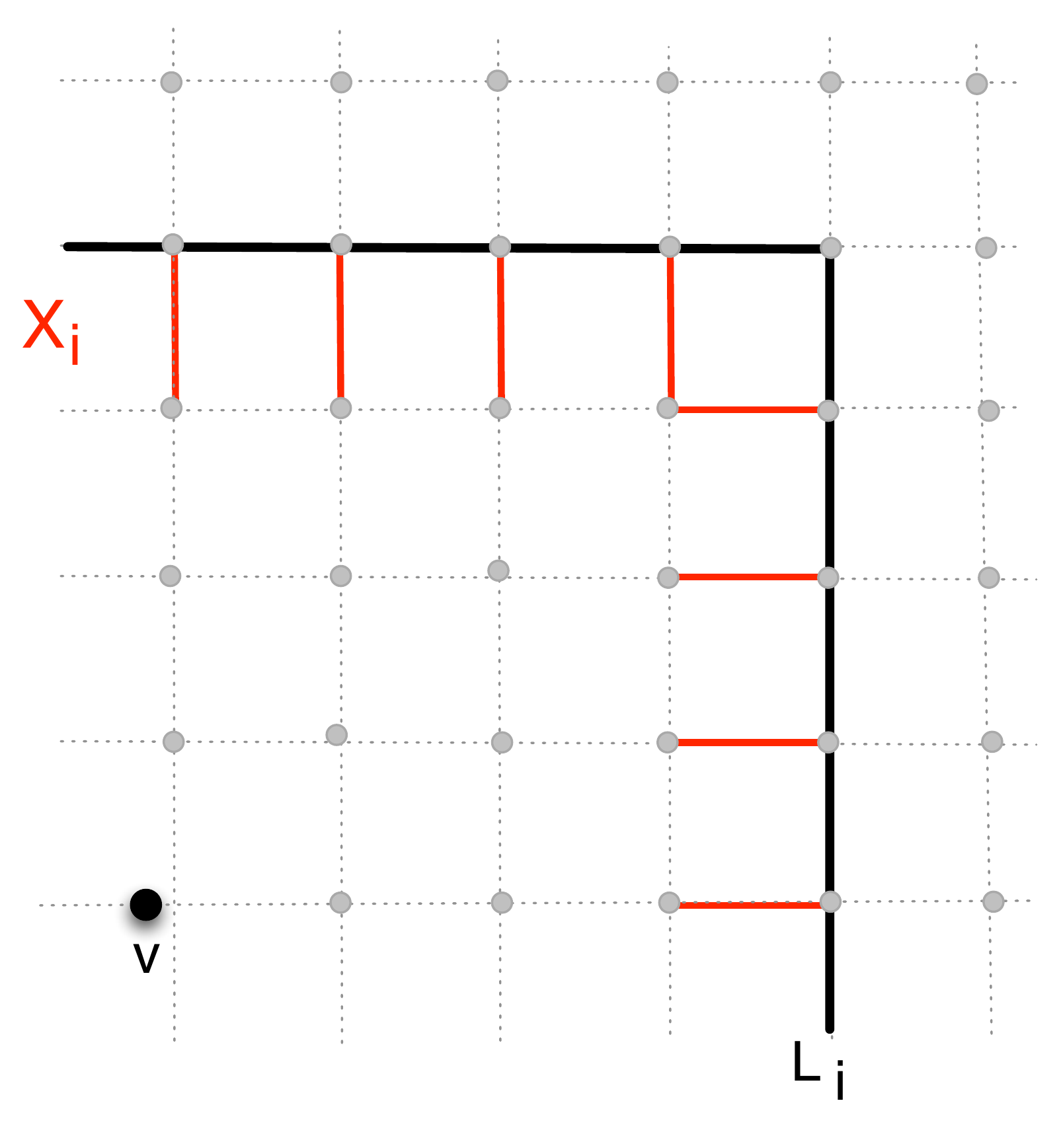}
  }
 \caption{The segments contributing to a layer's horizon in the gridlines model}
 \label{fig:gridlines} 
\end{figure}

\section{Experimental results}
\label{sec:results}

In this section we describe the implementation details and the results
of the experiments with our algorithms.
We implemented the five algorithms described above: \ioradII is the
layered radial sweep algorithm described in
Section~\ref{sec:radial-layers}; \ioradIII is the radial sweep
algorithm from Section~\ref{sec:radial-sectors}; \cosweep is the
centrifugal sweep algorithm from~Section~\ref{sec:conc-sweep};
\visiter and \visdac are the two algorithms described in
Section~\ref{sec:visdac} and \ref{sec:visiter}, respectively.
We use as reference the algorithm from our previous work,
\ioradI\cite{havertoma:visibility-journal}, which is also based on Van
Kreveld's model.

\subsection{Implementation details} 
We start by reviewing the implementation details of our algorithms. 

\ioradI scans the elevation grid, creates 3 events for each cell and
writes them to the event stream; with 12 bytes per event, the event
stream is $36n$ bytes.  It then sorts the event stream by azimuth
angle.  To sweep, it scans the event stream in order while using an
active structure to keep track of the events that intersect the sweep
line. During the sweep, the cells that are visible are written to a
file. In the end, the file is sorted by location and written to the
output grid. \ioradI can function in a recursive mode if it determines
that the active structure does not fit in memory. However in all
datasets the active structure is small ($<$30~MiB) and completely fits
in memory~\cite{havertoma:visibility-journal}.

\ioradII performs 2 passes over the elevation grid. It first maps the
elevation grid file in (virtual) memory and creates the sorted layers
$E_k$.  During this phase, the elevation grid and the sorted arrays
$E_k$ are kept in memory, and \ioradII relies on the virtual memory
manager (VMM) to page in blocks from the elevation grid as necessary
when accessing the points in band $k$.  Note that the accesses to the
elevation grid are not sequential (although they amount to
$O(\scan(n))$ \ios).  To help the VMM we implemented the following
strategy: whenever the current band needs to access an elevation from
the grid, we load an entire square tile of $\Theta(M)$ points in
memory, and keep track of the two most recent tiles.  Once all $E_k$
are computed the elevation grid is freed.  During the second phase
(the sweep) the elevations are accessed sequentially from the bands
$E_k$ and the output grid is kept in (virtual) memory as a bitmap
grid.

\ioradIII also performs 2 (sequential) passes over the elevation
grid. The first pass scans over the elevation grid and places each
point $(i,j, Z_{ij})$ in its sector. Sectors are stored as streams on
disk.  The second pass sorts and sweep the points in one sector at a
time by $\enterevent(i,j)$. The output grid is kept in (virtual) memory as a
bitmap grid. Except for the output grid, \ioradIII does not use the
VMM.

\ioradI, \ioradII and \ioradIII use the same data structures: a heap
as a priority queue; a red-black tree for the active
structure~\cite{cormen:introduction}; and the same in-memory sorting
(optimized quicksort).  For the largest terrains the priority queue
and the active structure are at most 30MB and fit in memory.

\cosweep is implemented in one (non-sequential) pass over the
elevation grid.
The implementation is recursive, as described in
Section~\ref{sec:conc-sweep}. Theoretically the algorithm could run
completely cache-obliviously with help of the VMM, but this turned out
to be slow. Therefore we implemented a cache-aware version: whenever
the recursion enters a tile $G$ of the largest size that fits in
memory, we load the elevation values for the entire tile into memory;
when the algorithm returns from the recursive call on $G$, the
visibility values for $G$ are written to disk.

\visiter performs two sequential passes over the elevation grid, and
one over the visibility grid.  The first pass reads the elevation grid
and creates the bands, and the second pass loads the bands one by one
in memory and computes the visibility bands. The horizon is maintained
as an array of (azimuth, zenith) pairs, and is accessed sequentially; in all
our experiments it never exceeded 200,000 points.

\visdac implements a divide-and-conquer refinement of \visiter. When
when a band is loaded in memory, \visdac will compute and merge the
layers in the band in a way similar to mergesort, which leads to an
improved upper bound for its time complexity.
\visdac and \visiter share the same code. The user can switch between
the two by turning on or off a flag that triggers the
divide-and-conquer. There is another flag to select the model
(gridlines or layers).

The implementations of all of our algorithms avoid taking square roots
and arctangents, and do not store any angles.  Instead of elevation
angles, they use the signed squared tangents of elevation angles, and
instead of azimuth angles, they use tangents of azimuth angles
relative to the nearest axis direction (north, east, south, or west).

\subsection{Platform} The algorithms are implemented in C and compiled
with gcc/g++ 4.1.2 (with optimization level -O3). All experiments were
run on HP 220 blade servers, with an Intel 2.83 GHz processor and a
5400 rpm SATA hard drive (the HP blade servers come only with this HD
option). The machine is quad-core, but only one CPU was used. 
We ran experiments rebooting the machine with 512 MiB and 1 GiB of
RAM. These sizes do not reflect current technology, and  have been
chosen in order to emphasize the scalability of the algorithms to a
volume of data that is much larger than the amount of RAM available.

\subsection{Datasets}
The algorithms were tested on real terrains ranging up to over 7.6
billion elements, see~Table~\ref{tbl:datasets} for some examples.  The
largest datasets are SRTM1 data, 30m resolution, available at
\texttt{http://www2.jpl.nasa.gov/srtm/}.
We selected these datasets because they are easily available, and are
large. In practice it does not make sense to compute a viewshed on a
very large area at low resolution; instead we would want to use a grid
corresponding to a relatively small area at high resolution. We used
SRTM data because it was easily and freely available, and served our
goal to compare the algorithms.


On all datasets up to 4~GiB (\texttt{Washington}), viewshed timings
were obtained by selecting several viewpoints uniformly on each
terrain and taking the average time.
For the larger datasets we chose a viewpoint approximately in the
middle of the terrain. This gives a good indication of the algorithm's
performance, and we note that for all our algorithms the majority of
the running time is spent handling the bands, and we expect that the
total runing time will vary insignificantly with the position of the
viewpoint.

\begin{table*}[t]
  \centering
  \begin{tabular}{|r|r@{ $\times$}rr|}
    \hline
   Dataset &  \multicolumn{3}{c|}{Size} \\
           & cols & rows & GiB \\
     \hline
     {Cumberlands}      &  8\,704 &  7\,673 &  0.25 \\
     {USA DEM 6}        & 13\,500 & 18\,200 &  0.92 \\
     {USA DEM 2}        & 11\,000 & 25\,500 &  1.04 \\
     {Washington}       & 31\,866 & 33\,454 &  3.97 \\
     {SRTM1-region03}   & 50\,401 & 43\,201 &  8.11 \\
     {SRTM1-region04}   & 82\,801 & 36\,001 & 11.10 \\
     {SRTM1-region06}   & 68\,401 &111\,601 & 28.44 \\
     \hline
  \end{tabular}
 \caption{Sample datasets}
  \label{tbl:datasets}
 \end{table*}

\subsection{Results}
Figure~\ref{fig:512mb} shows the total running times of all our
algorithms with 512~MiB RAM.
%
First, we note that all our algorithms, being based on I/O-efficient
approaches, are scalable to data that is more than sixty times larger
than the memory of the machine. This is in contrast with the
performance of an internal-memory algorithm, which would start
thrashing and could not handle terrains moderately larger than memory,
as showed in~\cite{havertoma:visibility-journal}.

\ioradI, \ioradII and \ioradIII are all based on radial sweeps of the
terrain, and theoretically they all use $\Theta(\sort(n))$~\ios. In
practice, however, both our new algorithms are significantly faster
than \ioradI. On \textrm{Washington} (3.97~GiB), \ioradI runs in 32,364
seconds with 16\% CPU\footnote{The numbers for \ioradI are different
than the ones reported in~\cite{havertoma:visibility-journal}
because the current platform has a slower disk.}, while \ioradII
runs in 13,780 seconds (22\% CPU), and \ioradIII in 3,009 seconds
(89\% CPU). This is a speed-up factor of more than 10.

\begin{figure}[t]
\centering {
\includegraphics[height=1.7in]{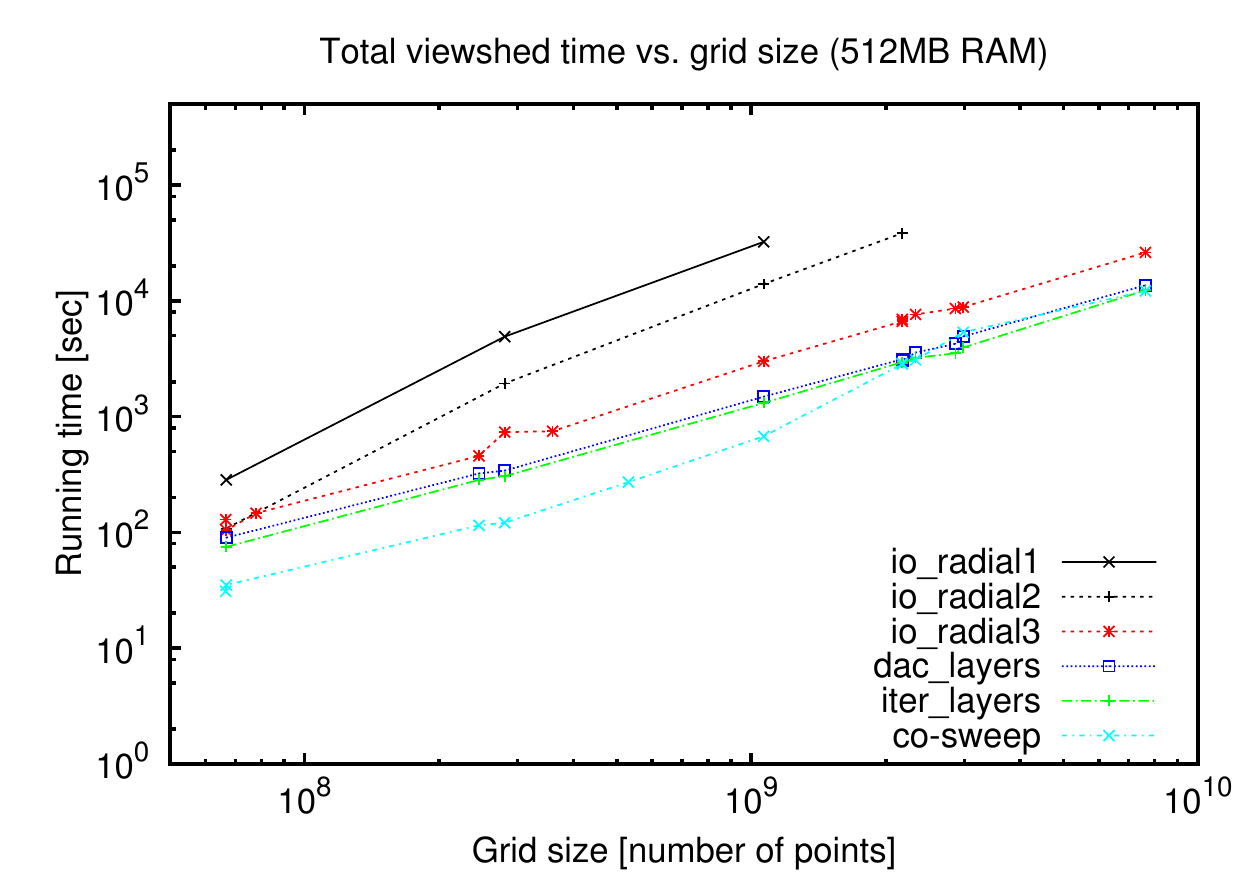}
\includegraphics[height=1.7in]{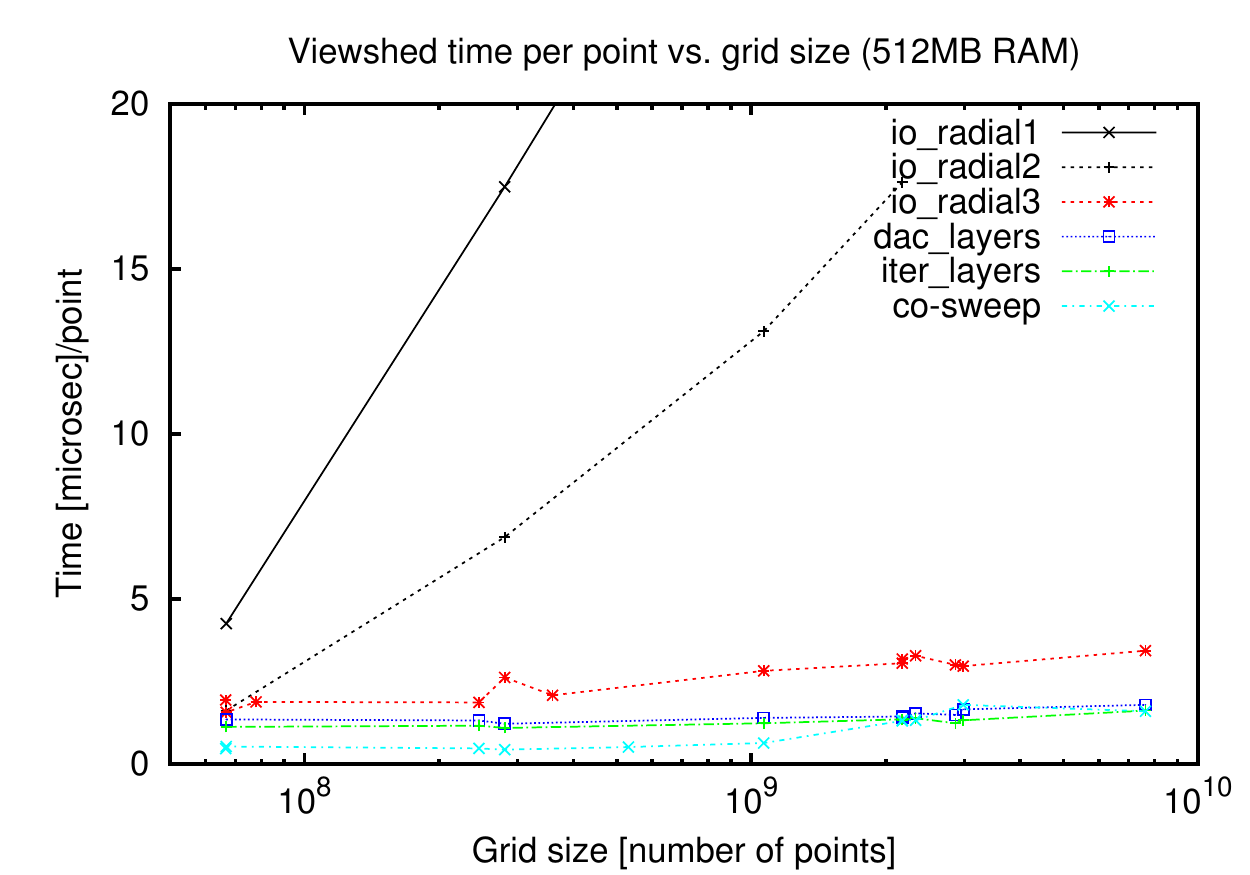}
}
\caption{Running times with 512MB RAM.  (a) Total time (log scale)
 with dataset size (log scale). (b) Total time per point.}
\label{fig:512mb}
\end{figure}




Both \ioradII and \ioradIII perform two passes over the elevation
grid, however \ioradIII is much faster. On a machine with 512~MiB RAM,
on \texttt{SRTM-Region 3} (8.11~GiB), \ioradII takes 37,982 seconds
(16\% CPU), while \ioradIII runs in 6,644 seconds (81\% CPU). Overall,
for \ioradII roughly 20\% of the time is CPU time, while for \ioradIII
the CPU utilization is 80\% or more.
%
The difference may be explained by the fact that the first
pass of \ioradII is non-sequential (although it performs $O(n/B)$
\ios), while both passes of \ioradIII are sequential. Another
difference is that \ioradII uses the VMM more than \ioradIII. 
%

Our third algorithm, \cosweep, is the fastest. It finishes a
28.4 GiB terrain (\texttt{SRTM-Region 6}) in 12,186 seconds (203
minutes), while \ioradII takes 26,193 seconds (437 minutes).  For
\ioradIII, 61\% of this time is CPU time, while for \cosweep only 18\%.
%
The reason is that \cosweep does a single pass through the elevation
grid. For any grid point $u$, the highest elevation angle of the
$O(\sqrt n)$ cells that may be on the line of sight from $v$ to $u$ is
retrieved from the horizon array in $O(1)$ time, and the horizon array
is maintained in $O(1)$ time per point on average. As a result,
\cosweep is CPU-light and the bottleneck is loading the blocks of data
into memory. \ioradIII, on the other hand, is more computationally
intensive---the highest elevation angle on the line of sight to $u$
needs to be retrieved from a red-black tree in $O(\log n)$ time, and
that tree is maintained in $O(\log n)$ time per point. In addition,
\ioradIII needs time to sort events.

One of our findings is that relying purely on VMM, even for a
theoretically I/O-efficient data access, is slow.  The analysis of the
I/O-efficiency of both \ioradII and \cosweep is based on the
assumption that the VMM will automatically load tiles of size
$\Theta(M)$ into memory in the optimal way, and that in practice the
performance will not be very different (the theoretical foundations
for this assumption were given by~\cite{prokop:cob}). In practice this
did not work out so well: a fully cache-oblivious, VMM-based
implementation of \cosweep and \ioradII turned out to be slow. By
telling the algorithms explicitly when to load a memory-size block
(and not using the VMM), we obtained significant speedups (without
sacrificing I/O-efficiency for the levels of caching of which the
algorithm remained oblivious, and without sacrificing CPU-efficiency).
We believe that we could further improve the running time of \cosweep
by having it manage the process of caching in memory the blocks of
data that it needs to access from the grid on disk (write its own
block manager with LRU policy).

Interestingly, our linear interpolation algorithms \visiter and
\visdac are faster than \ioradIII, and slightly slower than \cosweep.
For all datasets we found that $n \leq cM^2/B$ for a sufficiently
small constant $c$, which means that \textsc{BuildBands} and
\textsc{CollectBands} run in a single pass over the data. Thus both
algorithms perform two passes over the elevation grid, and a pass over
the visibility grid to assemble to visibility bands.  The total
running time is split fairly evenly between the three phases. The
actual visibility calculation runs at 100\% CPU and represents $<$
25\% of the running time.  More than 75\% of the total time is spent
reading or writing the bands.

In all our tests we found that the iterative algorithm \visiter is
consistently 10-20\% faster than \visdac.
%
To understand this we investigated the size of the horizons computed
by \visdac and \visiter: $H_i$, the horizon of layer $i$; and
$H_{1,i}$, the horizon of the points in the first $i$ layers.  Note
that the number of grid points on level $i$ is $8i$, and the total
number of points on levels $1$ through $i$ is $4i^2 + 4i+1 =
\Theta(i^2)$.  We know that $H_i = O(i) = O(\sqrt n)$, and $H_{1,i} =
O(i^2) = O(n)$ (Theorem~\ref{thm:gridhorizon}).
We recorded $|H_i|$ and $|H_{1,i}|$ for each layer $i$ during the
execution of \visiter.  Figure~\ref{fig:horizons} shows the results
for two datasets;
the results for the other datasets look similar.

\begin{figure*}[htb]
  \centering{
  \includegraphics[width=6cm]{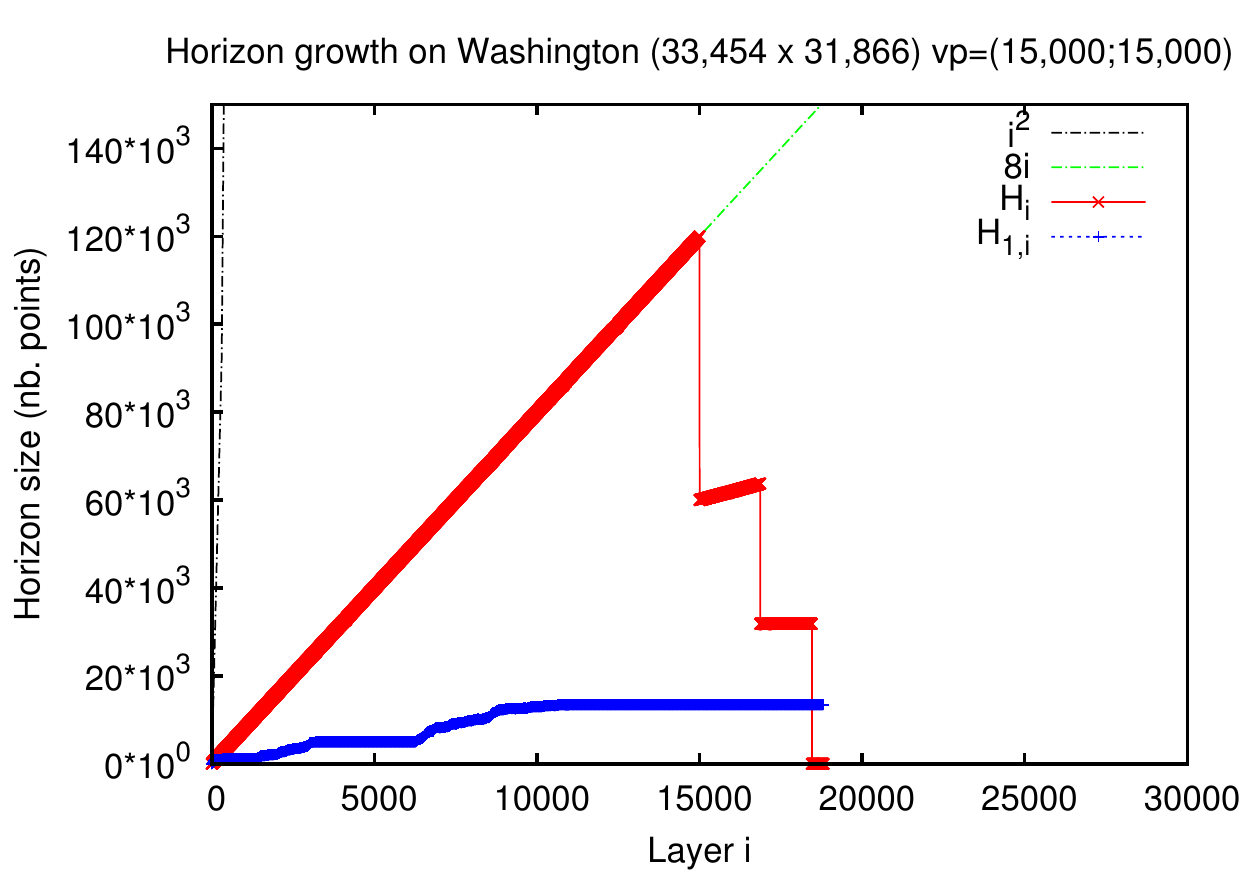}
  \includegraphics[width=6cm]{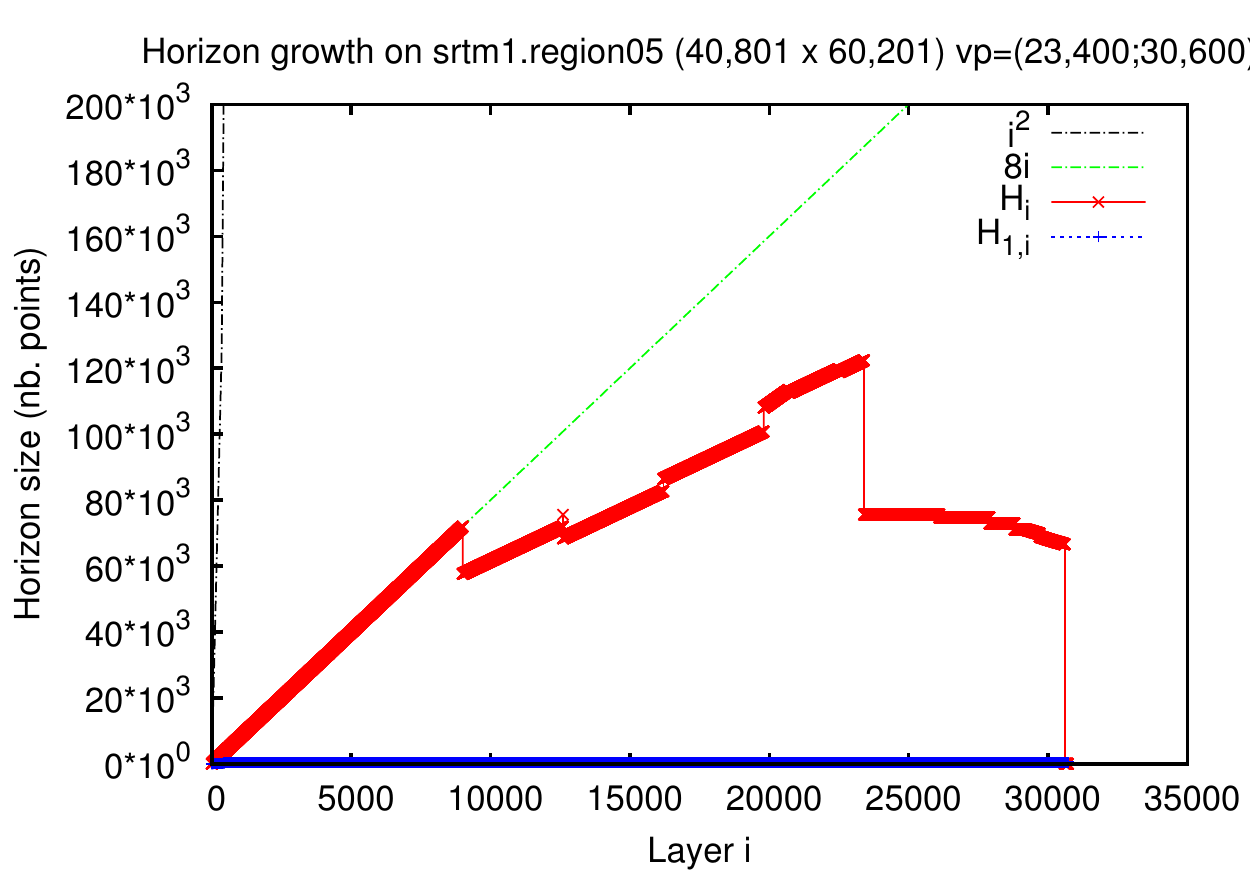}
  }
  \caption{Horizon growth for a viewshed computation on datasets
    \texttt{Washington} and \texttt{srtm1.region05} }
  \label{fig:horizons}
\end{figure*}

We see that $|H_i|$ is very close to its theoretical bound of $8i =
\Theta(i)$. As $i$ gets larger the $i$th layers starts to fit only
partially in the grid, and this causes $|H_i|$ to drop and have steep
variations.  The main finding is that for all datasets $H_{1,i}$ stays
very small, far below its theoretical upper bound of $\Theta(i^2)$.
$H_{1,i}$ grows fast initially
and then flattens out; For e.g. on \texttt{Washington} dataset
(approx. 1 billion points), $|H_{1,i}|$ flattens at 13,452 points; and
on \texttt{srtm1.region05} (approx. 2.5 billion points), $|H_{1,i}|$
flattens at 460 points.  All SRTM datasets have the horizon
$H_{1,O(\sqrt n)}$ between 132 and 32,689.

\begin{figure*}[t]
  \centering{
    \includegraphics[width=6cm]{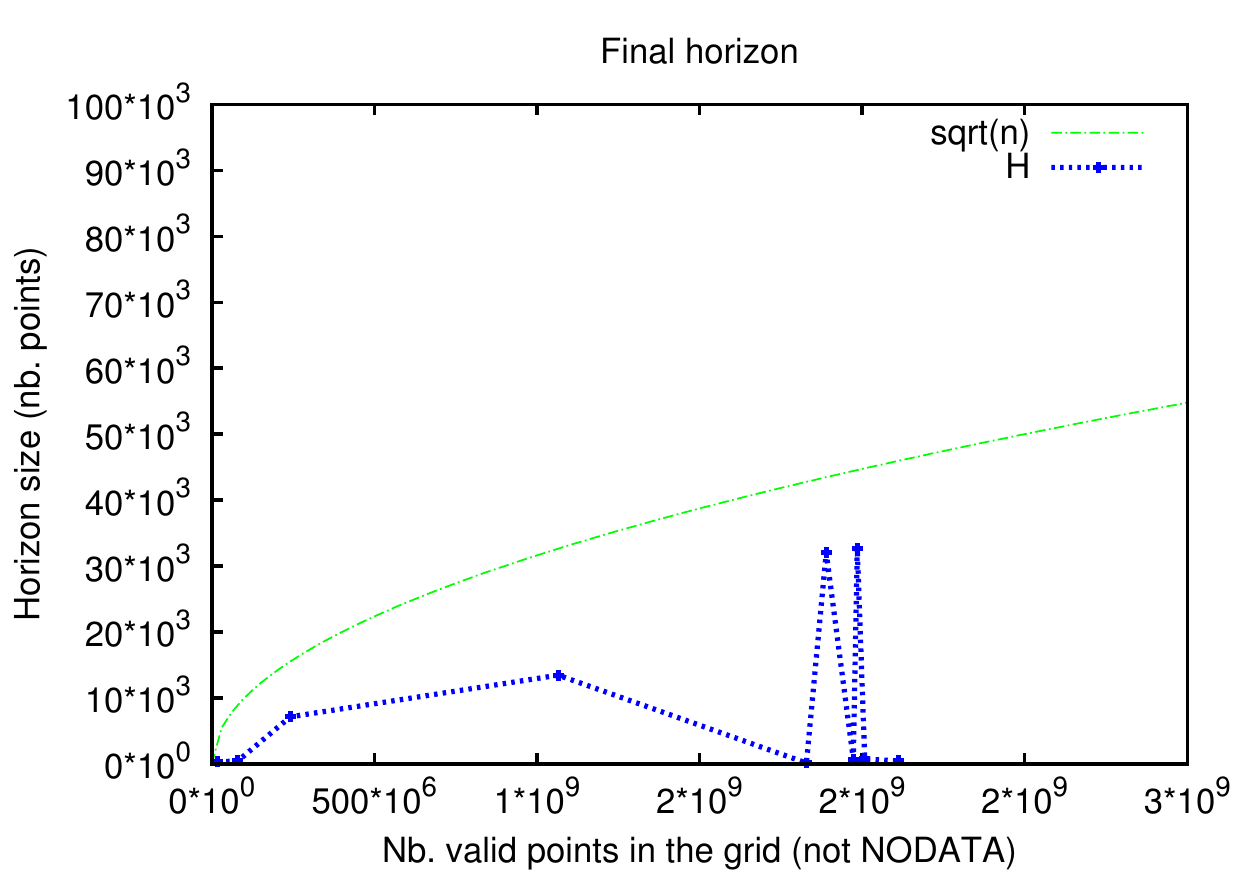}
    \includegraphics[width=6cm]{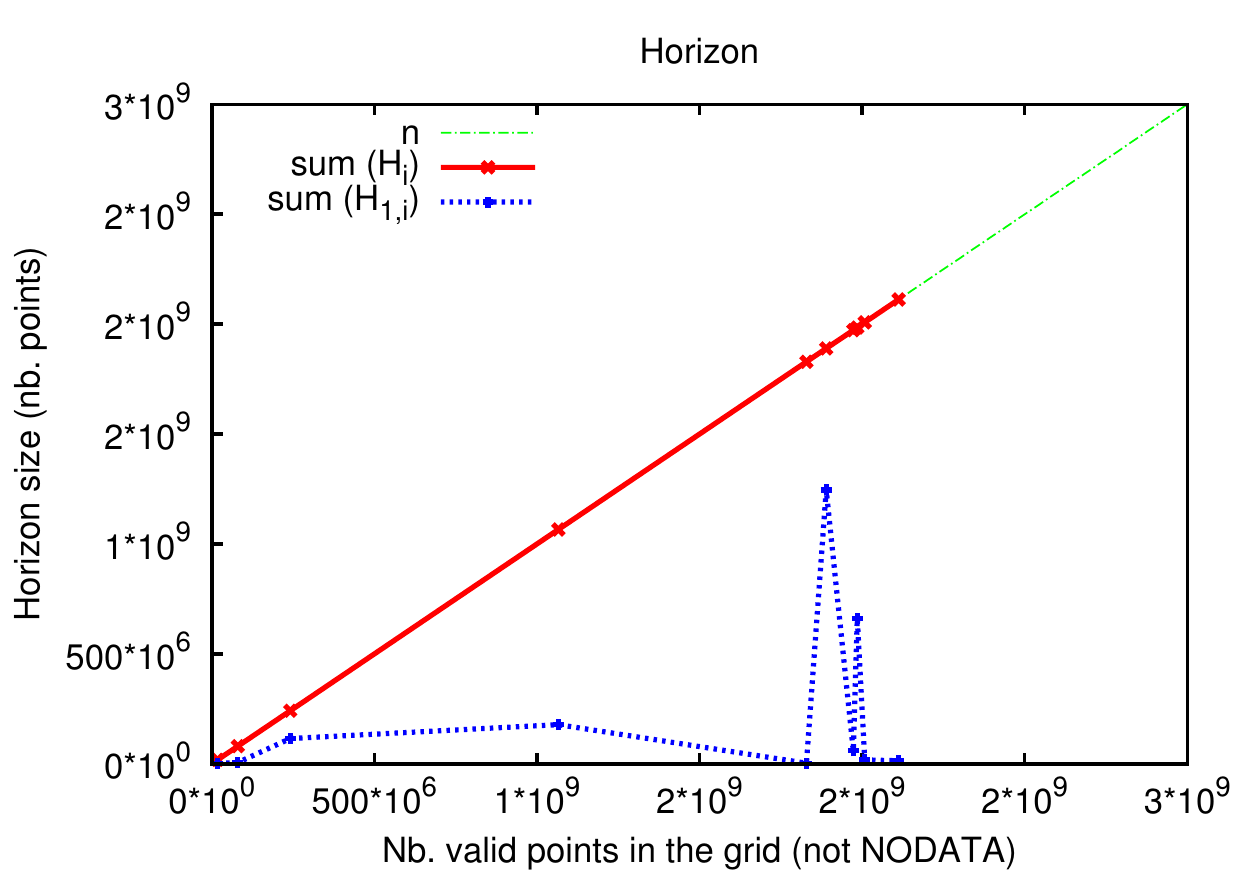}
  }
  \caption{(a) Size of final grid horizon $ |H_{1,O(\sqrt n)}|$ with dataset
    size. (b) $\sum_{i=1}^{O(\sqrt n)} |H_{1,i}|$ and $\sum_{i=1}^{O(\sqrt n)} |H_i|$ with dataset size.}
  \label{fig:finalhorizon}
\end{figure*}

Given a dataset, we refer to the horizon $ H_{1,O(\sqrt n)}$ as its
\emph{final} horizon. Figure~\ref{fig:finalhorizon}(a) shows the size
of the final horizon for each dataset as function of the number of
\emph{valid} points in the grid ---- this excludes the points in the
grid that are labeled as \emph{nodata}, and which are used for e.g. to
label the water/ocean; these points do not affect the size of the
horizon, as chains of \emph{nodata} points are compressed into a
single horizon segment.  We see that the final horizon: (1) has a lot
of variation especially for the larger SRTM datasets, jumping from low
to high values. This is likely due to the position of the viewpoint
and possibly the topology of the terrain; (2) the horizon stays small,
below $\sqrt n$ for all datasets, far below its worst-case bound of
$O(n)$.

Figure~\ref{fig:finalhorizon}(b) shows the cumulative sums,
$\sum_{i=1}^{O(\sqrt n)} |H_{1,i}|$ and $\sum_{i=1}^{O(\sqrt n)} |H_i|$, for
each dataset, as a function of the number of valid points in the grid;
we recorded these sums because they come up in the analysis of
\visiter and can shed light on its performance. In
Figure~\ref{fig:finalhorizon} we see that $\sum_{i=1}^{O(\sqrt n)} |H_i|$ grows indeed
linearly with the number of valid points in the grid. The sum
$\sum_{i=1}^{O(\sqrt n)}
|H_{1,i}| $ has a lot of variability similar with the
final horizon shown in Figure~\ref{fig:finalhorizon} (a), and for all
datasets stays far from its worst-case upper bound  of $O(n \sqrt n)$.
We note that Figure~\ref{fig:horizons} and \ref{fig:finalhorizon} are
based on a single viewpoint, but we expect the results will carry
over.

Comparing to the work of Ferreira et al.~\cite{ferreira:vis}:  Their
 algorithm, \textsc{TiledVS}, also consists of three passes: convert
 the grid to Morton order, compute visibility using R2 algorithm, and
 convert the output grid from Morton order to row-major order.
They report on the order of 5,000 seconds with \textsc{TiledVS} for
\texttt{SRTM1.region06}, using a similar platform as ours and
additional optimization like data compression.
Assuming that this time includes all three passes, and modulo
variations in setup, it is approx. 2.5 times faster than \visiter.  
We note that \textsc{TiledVS} uses a different model and we believe
our algorithms and their analysis  are of independent interest.

\section{Conclusion}\label{sec:discussion}

In this paper we described new I/O-efficient algorithms for computing
the visibility map of a point on a grid terrain using several
different models.  The algorithms are provably efficient in terms of
the asymptotic growth behaviour of the number of I/Os, but at the same
time are designed to exploit that the terrain model is a grid.  This
leads to much improved running times compared to our previous
work~\cite{havertoma:visibility-journal}. On the largest terrains,
using as little as 512~MiB of memory, our algorithms perform at most
two passes through the input data, and one pass through the output
grid. We were able to compute viewsheds on a terrain of 28.4~GiB in
203 minutes with a laptop-speed hard-drive. The algorithms that
compute (what is considered to be) the exact viewshed have inferior
worst-case upper bounds, but in practice are faster than the
radial-sweep algorithms due to the small size of horizons. We conclude
that horizon-based algorithms emerge as a fast approach for computing
viewsheds.

As avenues for future reesarch we mention the problem of proving a
sub-linear bound for the expected complexity of a horizon, and
obtaining an output-sensitive viewshed algorithm.

\section{Acknowledgments}
The authors thank DJ Merill for setting up and administering the
platform used for the experiments. And former Bowdoin students Jeremy
Fishman and Bob PoFang Wei for working on the earlier versions of this
paper.

\bibliographystyle{acmtrans}
\bibliography{viewshed}

\end{document}